%% file: main.tex
\PassOptionsToPackage{dvipsnames}{xcolor}
\documentclass[acmsmall,nonacm]{acmart}

\usepackage{listings}
\usepackage{mathpartir}
\usepackage{subcaption}
\usepackage{wasysym}
\usepackage{tikzit}
\usepackage{graphicx}
\input{figures/sample.tikzstyles}
\input{figures/sample.tikzdefs}
\usepackage{pgf}
\usepackage{thmtools}
\usepackage{siunitx}
\usepackage[utf8]{inputenc}

\newcommand{\commentNote}[1]{}
\newcommand{\Lang}{\textsc{DeCo}}

\newcommand{\den}[1]{\llbracket #1 \rrbracket}

\newcommand{\pos}[1]{\llbracket #1 \rrbracket_\mathcal{I}}
\newcommand{\deriv}[1]{{#1}_\mathsf{d}}
\newcommand{\ctype}[1]{{#1}_\mathsf{C}}
\newcommand{\init}[1]{{#1}_\mathsf{i}}

\newcommand{\src}[1]{{\color{RoyalBlue}{\mathsf{#1}}}} %
\newcommand{\ext}[1]{{\color{RedOrange}{\mathbf{#1}}}} %
\newcommand{\der}[1]{{\color{ForestGreen}{\mathit{#1}}}} %

\newcommand{\pair}[2]{\src{\langle} #1 \src{,}~ #2 \src{\rangle}}

\usepackage{color}
\definecolor{keywordcolor}{rgb}{0.7, 0.1, 0.1}   %
\definecolor{tacticcolor}{rgb}{0.0, 0.1, 0.6}    %
\definecolor{commentcolor}{rgb}{0.4, 0.4, 0.4}   %
\definecolor{symbolcolor}{rgb}{0.0, 0.1, 0.6}    %
\definecolor{sortcolor}{rgb}{0.1, 0.5, 0.1}      %
\definecolor{attributecolor}{rgb}{0.7, 0.1, 0.1} %

\usepackage{listings}

\newcommand{\sound}{value preserving}
\newcommand{\soundness}{value preservation}
\newcommand{\Sound}{Value Preserving}
\newcommand{\Soundness}{Value Preservation}
\newcommand{\thmIncrSoundName}{\Sound{} Incrementalization}
\newcommand{\thmCombSoundName}{\Sound{} Combinators}

\setcopyright{rightsretained}
\acmDOI{10.1145/3798264}
\acmYear{2026}
\acmJournal{PACMPL}
\acmVolume{10}
\acmNumber{OOPSLA1}
\acmArticle{156}
\acmMonth{4}
\received{2025-10-10}
\received[accepted]{2026-02-17}

\begin{document}

\title{DeCo: A Core Calculus for Incremental Functional Programming with Generic Data Types}

\author{Timon Böhler}
\orcid{0009-0002-9964-7367}
\affiliation{%
  \institution{TU Darmstadt}
  \city{Darmstadt}
  \country{Germany}
}
\email{timon.boehler@tu-darmstadt.de}

\author{Tobias Reinhard}
\orcid{0000-0003-1048-8735}
\affiliation{%
  \institution{TU Darmstadt}
  \city{Darmstadt}
  \country{Germany}
}
\email{tobias.reinhard1@tu-darmstadt.de}

\author{David Richter}
\orcid{0000-0002-8672-0265}
\affiliation{%
  \institution{TU Darmstadt}
  \city{Darmstadt}
  \country{Germany}
}
\email{david.richter@tu-darmstadt.de}

\author{Mira Mezini}
\orcid{0000-0001-6563-7537}
\affiliation{%
  \institution{TU Darmstadt}
  \city{Darmstadt}
  \country{Germany}
}
\affiliation{%
  \institution{hessian.AI}
  \city{Darmstadt}
  \country{Germany}
}
\affiliation{%
  \institution{National Research Center for Applied Cybersecurity ATHENE}
  \city{Darmstadt}
  \country{Germany}
}
\email{mezini@informatik.tu-darmstadt.de}

\begin{abstract}
\input{abstract}
\end{abstract}

\received{20 February 2007}
\received[revised]{12 March 2009}
\received[accepted]{5 June 2009}

\maketitle

\section{Introduction}

\input{introduction}

\input{byexample}

\section{Formalizing Incrementalization}
\label{sec:fw}

\input{framework}

\section{Soundness}
\label{sec:correct}

\input{correctness}

\section{Implementation}
\label{sec:impl}

\input{implementation}

\section{Applications}
\label{sec:app}

\input{application}

\section{Related Work}
\label{sec:rel}

\input{related}

\section{Conclusion}
\label{sec:conc}

\input{conclusion}

\section{Data-Availability Statement}
\input{data}

\section{Acknowledgments}

We thank the anonymous reviewers for their extensive and helpful feedback.

This research work was supported by the National Research Center for Applied Cybersecurity ATHENE. ATHENE is funded jointly by the German Federal Ministry of Education and Research and the Hessian Ministry of Higher Education, Research and the Arts.
This work was funded by the LOEWE initiative (Hesse, Germany) [LOEWE/4a//519/05/00.002(0013)/95], by the Deutsche Forschungsgemeinschaft (DFG, German Research Foundation) – SFB 1119 – 236615297, by the Hessian Ministry of Higher Education, Research, Science and the Arts within the cluster project The Third Wave of Artificial Intelligence (3AI), and by the Deutsche Forschungsgemeinschaft (DFG, German Research Foundation) under Germany’s Excellence Strategy (EXC-3057/1 ``Reasonable Artificial Intelligence'', Project No. 533677015).

\bibliographystyle{ACM-Reference-Format}
\bibliography{bib}

\newpage
\appendix

\section{Full Listings}
\label{sec:sums}
\input{sums}

\clearpage

\section{Evaluation Details}
\label{sec:evaldetails}
\input{evaldetails}

\end{document}

%% file: figures/sample.tikzstyles
\tikzstyle{red dot}=[fill=red, draw=black, shape=circle]
\tikzstyle{green dot}=[fill={rgb,255: red,0; green,113; blue,0}, draw=black, shape=circle]
\tikzstyle{medium box}=[fill=white, draw=black, shape=rectangle, minimum width=0.5cm, minimum height=0.5cm]
\tikzstyle{black dot}=[fill=black, draw=black, shape=circle, inner sep=0pt, outer sep=0pt, minimum size=0.15cm]

\tikzstyle{dashed edge}=[-, dashed, draw={rgb,255: red,255; green,128; blue,0}, fill=none]
\tikzstyle{big dashes}=[-, thick, dashed, dash pattern=on 4mm off 2mm, fill={rgb,255: red,128; green,0; blue,128}]
\tikzstyle{filled yellow}=[-, fill=yellow, draw=none]
\tikzstyle{filled orange}=[-, fill={rgb,255: red,255; green,128; blue,0}, draw=none]
\tikzstyle{arrowheadsarrowheadsarrowheads}=[-, ->, line width=1pt]
\tikzstyle{thickedge}=[-, line width=1pt]
\tikzstyle{bi arrow}=[-, <->]
\tikzstyle{teal line}=[-, draw={rgb,255: red,0; green,128; blue,128}]
\tikzstyle{red_line}=[-, draw={rgb,255: red,191; green,0; blue,64}]
\tikzstyle{red}=[-, fill=red, draw=none]

%% file: abstract.tex
Incrementalization speeds up computations by avoiding unnecessary recomputations and by efficiently reusing previous results.
While domain-specific techniques achieve impressive speedups, e.g., in the context of database queries, they are difficult to generalize.
Meanwhile, general approaches offer little support for incrementalizing domain-specific operations.

In this work, we present \Lang{}, a novel core calculus for incremental functional programming with support for a wide range of user-defined data types.
Despite its generic nature, our approach statically incrementalizes domain-specific operations on user-defined data types.
It is, hence, more fine-grained than other generic techniques which resort to treating domain-specific operations as black boxes.

We mechanized our work in Lean and proved it sound, meaning incrementalized execution computes the same result as full reevaluation.
We also provide an executable implementation with case studies featuring examples from linear algebra, relational algebra, dictionaries, trees, and conflict-free replicated data types,
plus a brief performance evaluation on linear and relational algebra and on trees.

%% file: introduction.tex
\newcommand{\dense}{\ensuremath{\mathsf{dense}}}

\commentNote{When an input $x$ changes by $\Delta x$, an incrementalized function leverages previously computed results to reduce computational overhead.
For example, consider a dense neural network layer function \dense{} defined as:
\[\dense{}~x := \mathsf{map}~\mathsf{relu}~(M \cdot x + b)\] 
where $x$ is the input vector, $M$ is the (square) weight matrix, and $b$ is the bias vector. 
The standard evaluation of \dense{} has $O(n^2)$ time complexity with respect to the input size. 
However, once we incrementalize it, only the initial computation will require $O(n^2)$ time for input vector $x$.
Any subsequent update will merely take $O(n)$ when modifying a single element of $x$.
 i.e., we achieve linear-time updates by efficiently reusing previous computations.}

Incrementalization~\cite{Liu00} speeds up computations by avoiding full recomputation and processing only input changes. 
We present \Lang{}, a typed functional first-order language for incremental programming that supports static, fine-grained incrementalization using a combination of user-defined data types and built-in sum and product types. 

\Lang{} addresses the limitations of and complements the state of the art as follows:
Unlike coarse-grained approaches such as selective reevaluation~\cite{Hammer14} that treat operations as black boxes, \Lang{} incrementalizes operations on user-defined data types while caching only relevant intermediate results.
Unlike fine-grained approaches targeting specific structures such as multidimensional arrays~\cite{Shaikhha20} or relations~\cite{Budiu23}, 
\Lang{} provides a unified foundation for incremental computation across different data types through \emph{containers}~\cite{Abbott05}. 
Data-type-agnostic approaches like I$\lambda$C~\cite{Cai14} provide little support when implementing new data types
as they have no generic built-in operations, while
\Lang{} offers data-type-generic operations such as $\src{map}$ and $\src{reshape}$ and incrementalization combinators
for deriving incrementalizations more easily.

\paragraph{Unified Treatment of Different Data Types}
To effectively increment operations on user-defined data structures, \Lang{} must exploit data-structure-specific properties, while also supporting a wide range of such structures.
Our calculus achieves this by representing type constructors as \emph{containers}~\cite{Abbott05} and requiring base types to form \emph{change structures}~\cite{Cai14}.
Containers offer a uniform representation of data structures as maps from indices to values, enabling simple in-place updates---in contrast to algebraic data types, which require recursively traversing unchanged parts of a structure.
Change structures provide a type of values, a type of changes, and operations applying changes to values.
Users can instantiate \Lang{} to their domain by defining (a)~a container for their data structure, (b)~a change structure for their base type, and (c)~domain-specific operations and their incrementalizations.
\Lang{} also comes with product and sum types built-in, automatically deriving change structures for arbitray combinations of base types,
sum, products, and containers.

\paragraph{Incrementalization Combinators}
Incrementalizing domain-specific operations often follows certain patterns and can be repetitive.
Yet, subtle mistakes may invalidate the entire computation.
\Lang{} reduces the user burden and strengthens user guarantees by providing a set of correct-by-construction incrementalization combinators adhering to these patterns.
They allow users to exploit a function's linearity or self-maintainability.
Most importantly, they guarantee soundness: 
Any domain-specific incrementalization built via our combinators yields the same result as the original computation.

\paragraph{Mechanized Soundness Proof}
We mechanize \Lang{} in the dependently-typed programming language Lean~\cite{Lean4}.
In particular, we mechanize our calculus' denotational semantics and incrementalization.
Most importantly, we mechanically prove that the incrementalization is sound with respect to the formalized denotational semantics.
That is, applying the incremental program to a sequence of changes yields the same result as recomputing the original program. Additionally, we prove that every incrementalization built with our incrementalization combinators is correct by construction, provided that each combinator's preconditions are met.

\paragraph{Implementation and Case Studies}
Building on the mechanization of our calculus, we implemented a generic incrementalization
framework, which we apply to incremental operations stemming from linear algebra, relational algebra, dictionaries, trees, and conflict-free replicated data types.
We thereby show how \Lang{}'s generic data transformations together with incrementalization combinators can be used to incrementalize operations from a variety of domains. 
Further, we incrementalize several programs, demonstrating the expected speed-up for incrementalization compared to full reevaluation:
A dense neural network layer, two relational queries involving join and aggregation, and a sum over rose trees.

\paragraph{Contributions}
In summary, we present \Lang{}, a novel core calculus for incremental
functional programming with the following key properties:
(i) Support for built-in sum and product types alongside a wide range of \emph{user-provided} data structures.
(ii)~Support for facilitating the instantiation to new application domains through generic operations and combinators.
(iii)~Automatic incrementalization that is correct-by-construction.
Concretely, we make the following contributions.
\begin{itemize}
	\item \textbf{Core calculus:} 
	A minimal formal foundation for cached incremental functional programming with generic data transformations, which can be instantiated to different data types.
	\item \textbf{Incrementalization combinators:} 
	Reusable building blocks that simplify the definition of correct-by-construction incrementalizations for user-defined operations.
	\item \textbf{Mechanized soundness proof:} 
	Incrementalization preserves program semantics.
	The proof is fully mechanized in Lean.
	\item \textbf{Implementation:} 
	A library that implements our formalized core calculus.
	\item \textbf{Case studies:} We instantiate our library to implement
		incremental operations involving trees, linear and relational algebra, and an exemplary conflict-free replicated data type (CRDT).
\end{itemize}

\paragraph{Paper overview}

The remainder of this paper is organized as follows. 
First, in Section~\ref{sec:ci}, we describe cached incrementalization and the correctness properties it should fulfill.
In Section~\ref{sec:ex}, we explain the different parts of our approach
using a dense layer program as the main example. We then formally define
\Lang{}, its type system, denotational semantics, and incrementalization,
as well as incrementalization combinators in Section~\ref{sec:fw}.
An overview of the mechanization and verification of these concepts
in Lean is given in Section~\ref{sec:correct}. The library implementation
and its application to our case studies are described in
Section~\ref{sec:impl} and \ref{sec:app}, respectively.
\footnote{A preliminary version of the artifact is available at \url{https://figshare.com/s/abdad549da026b4dbca5}.}

%% file: byexample.tex
\section{Cached Incrementalization in a Nutshell}
\label{sec:ci}

In this section, we introduce incrementalization, its cached variant, and the notion of correct incrementalization.
Incrementalization aims to speed up repeated computations on similar data by reusing previously computed results.
Our incrementalization approach is compositional and transforms regular programs into incremental ones, by mapping operations to incremental ones. %

A naive approach to incrementalization could work as follows:
Consider types $A$ and $B$ and let $A'$ and $B'$ be the types of \textit{changes} on $A$ and $B$, respectively (in our setting, one can imagine $A'$ as containing sparse data structures whose non-zero entries describe how to update the entries of values in $A$). 
We write $x \oplus x'$ to denote the application of a change $x' \in A'$ to a value $x \in A$.
Then, an incrementalization of $f:A \to B$ is a function $f' : A' \to B'$ that maps input changes to output changes:
\begin{align}
f~(x \oplus x') = f~x \oplus f'~x'.
\end{align}
Note that throughout this work, we denote variables relating to changes using a superscript ${}^\prime$, as we did with $A', x', f'$ above.
An incrementalization $f'$ fulfilling the above law allows us to reuse the previously computed value $f~x$ during the evaluation of $f~(x \oplus x')$.
To compute the latter, we only need to add the difference $f'~x'$ to $f~x$.
Note that the above formula assumes that the output change solely depends on the input change, as $f'$ only receives $x'$, not $x$.
Often, however, we rely on more information to compute the change.
The most direct solution would be to also pass $x$ to~$f'$. 
However, then $f'$ may need to reevaluate parts of $f$ to determine how intermediate results have changed. 
To avoid this reevaluation, $f'$ should receive a data structure containing intermediate results.
This insight forms the core of the \textit{caching} approach to incrementalization~\cite{Giarrusso19}.

In general, \textit{cached incrementalization} remembers intermediate values by equipping each operation in the program with a cache.
Hence, in order to incrementalize any function $f : A \to B$, we need to define three components (Figure~\ref{fig:interface}):
(i)~A \textit{cache type} $C : \mathsf{Type}$ describing which values are cached. 
(ii)~A \textit{difference function} $d : A' \to C \to B' \times C$ that takes an input change and the previous cache, and returns an output change and a new cache. 
(iii)~An \textit{initializing function} $i : A \to B \times C$, which produces both the actual result of the function and the first cache.
Afterwards the difference function $d$ can be called on this initial cache. 

Figure~\ref{fig:intuit} visualizes this approach.
Consider a value $a$, repeatedly modified via changes~$a_1', a_2', \dots$, and a function $f$ taking $a$ as input.
A non-incremental approach (Figure~\ref{fig:intuit:fullReeval}) has to reevaluate $f$ each time $a$ changes.
In contrast, the incremental approach (Figure~\ref{fig:intuit:cachedIncremental}) only fully evaluates $f$ exactly once:
In the very first step, $i$ is evaluated.
In every following step, we call the incremental function $d$ on a cache and on the input change $a'_i$ (instead of the full input $a_i$).
This allows us to avoid full reevaluation.
\begin{figure}
\begin{minipage}[b]{.33\textwidth}
\[\begin{array}{lll}
f &:& A \to B \\[1em]
C &:& \mathsf{Type} \\
i &:& A \to B \times C \\
d &:& A' \to C \to B' \times C \\
\end{array}\]

\vspace{2em}
\caption{Incrementalization interface.}
\label{fig:interface}
\end{minipage}%
\begin{minipage}[b]{.66\textwidth}
\begin{subfigure}[b]{.5\textwidth}
\centering
\ctikzfig{incr_intuition}
\caption{Full re-evaluation.}
\label{fig:intuit:fullReeval}
\end{subfigure}%
\begin{subfigure}[b]{.5\textwidth}
\centering
\ctikzfig{incr_intuition2}
\caption{Cached incremental evaluation.}
\label{fig:intuit:cachedIncremental}
\end{subfigure}
\caption{Overview.}
\label{fig:intuit}
\end{minipage}
\end{figure}

Note that $d$ can also be viewed as the transition function of a Mealy machine with input type~$A'$, output type~$B'$, and internal state~$C$. 
The fact that Mealy machines implement \textit{causal stream functions}~\cite{Ghica25} means we can also interpret incrementalizations as functions of type $\mathsf{Stream}~A' \to \mathsf{Stream}~B'$, i.e., as maps from streams of input changes to streams of output changes. 
This perspective connects the cache-based approach to incrementalization of \citet{Giarrusso19} with the stream-based approach of \citet{Budiu23}. 
We use the former, because the idea of updating a cache comes closer to actual implementations. 
In contrast, a causal stream function takes in the entire previous history of input changes to produce the current output change, which obviously is not usable for direct implementation. 
Further, the Mealy machine makes explicit which data is being stored in the cache.

Our cached incrementalizations have to comply with the following laws (also Definition~\ref{def:incrementalization}):
(We write $i_1~x$ for the first component of $i~x$, and analogously for
$i_2, d_1$, and $d_2$.)
\begin{align}
\forall x.    && i_1~x      &= f~x                    \tag{Law-1} &\label{align:law1}\\
\forall x~x'. && f~(x \oplus x')     &= f~x \oplus d_1~x'~(i_2~x)  \tag{Law-2} \label{align:law2}\\
\forall x~x'. && i_2~(x \oplus x') &= d_2~x'~(i_2~x). \tag{Law-3} \label{align:law3}
\end{align}

The intuition behind these laws is as follows:
(\ref{align:law1})~The initialization function's first component equals the original full evaluation.
(\ref{align:law2}) allows us to avoid full recomputation on a changed value~$f~(x \oplus x')$.
Given the previous step's cache, we can use the difference function $d$ to compute the change.
(\ref{align:law3})~Applying the difference function $d$ to a change $x'$ and the initial cache produced by $i$ (i.e. $i_2~x$) updates the cache consistently.
In particular, initializing the cache for an updated input $x \oplus x'$ yields the same result as updating the existing cache with difference function~$d$.

The laws are used for showing our key soundness property.
In our soundness proof in Section~\ref{sec:correct} we show that our incrementalizations compute the same results as full reevaluations. More precisely, given a program $f$, initial value $x$ and a list of changes $x'_{(.)}$, we obtain the same result from:
(i)~Applying changes $x'_{(.)}$ to~$x$ and calling $f$ on the new value (i.e., full reevaluation).
(ii)~Computing $i~x$ and applying $d$ to each change (i.e., incrementalization).
This proof ensures \Lang{} programs support correct-by-construction incrementalization, enabling its use as foundational calculus and intermediate~representation~(cf.~Section~\ref{sec:frontend}).

\ref{align:law1} -- \ref{align:law3} capture the properties ensuring this soundness. They essentially describe the invariant that after each change, the cache always needs to align with the current output we would get from reevaluation.
To illustrate the idea, consider an initial input $a$ and a sequence of input changes $a_1', \dots, a_n'$ transforming $a$ into $a_n$.
Applying $d$ sequentially to each change yields a corresponding sequence of output changes $b_1', \dots, b_n'$ (cf. Figure~\ref{fig:intuit}).
We then apply these changes consecutively to the initial output $f~a$, thereby obtaining the updated output $f~a_n$.

Our use of intrinsic static typing means we only have to consider well-typed and terminating programs, which leads to a simpler correctness proof 
compared to untyped approaches~\cite{Giarrusso19}.
The above laws follow from relatively straightforward induction proofs.

\section{\Lang{} Programs by Example}
\label{sec:ex}

In this section, we
(i)  discuss the features of our core calculus \Lang{},
(ii) illustrate the support for domain-specific instantiations by implementing one for linear algebra, and
(iii) discuss the requirements that cached incrementalization imposes on the instantiation's domain-specific operations.
Later in Section~\ref{sec:app}, we showcase further instantiations which incrementalize relational algebra and conflict-free replicated data types.

In \Lang{} programs, we distinguish between:
(i)~The $\src{core~calculus}$ providing generic primitive operations applicable to all data structures.
(ii)~The $\ext{domain\text{-}specific}$ operations, which are tied to a specific data structure and provided during instantiation.
(iii)~$\der{Derived}$ operations are defined by the user via core and domain-specific operations, and \Lang{} automatically derives incrementalizations for them.\footnote{Following \citet{Patrignani20}, we distinguish the three kinds by color and use of sans serif, bold face, and italics, respectively.}

\paragraph{Containers}
\Lang{} represents data structures as containers~\cite{Abbott05}. A container is defined
a set of shapes, and, for each shape, a set of indices. Formally:
\[\begin{array}{lll}
S &:& \mathsf{Type} \\
\pos{\cdot} &:& S \to \mathsf{Type} \\ 
\end{array}\]
For example, arrays are defined as the container where the shapes are natural numbers (representing the length of the array) and where a shape $n$ has indices $0, \dots, n-1$. Hence, for arrays, we have
$\pos{n} := \{0, \dots, n-1\}$.
Our language's semantics is defined over an arbitrary such container.

\paragraph{Generic operations}

\begin{figure}
\begin{minipage}{.55\textwidth}
\[\begin{array}{lll}
\src{dup} &:& A \leadsto A \times A \\
\src{cst} &:& \den{A} \to (B \leadsto A) \\
\src{map} &:& (A \leadsto B) \to (\src{F}_s~A \leadsto \src{F}_s~B) \\
\src{zip} &:& \src{F}_s~A \times \src{F}_s~B \leadsto \src{F}_s~(A \times B) \\
\src{map2} &:& (A \times B \leadsto C) \to (\src{F}_s~A \times \src{F}_s~B \leadsto \src{F}_s~C) \\
\src{get} &:& \pos{s} \to (\src{F}_s~A \leadsto A) \\
\src{set} &:& \pos{s} \to (A \times \src{F}_s~A \leadsto \src{F}_s~A) \\
\src{reshape} &:& (\pos{s_2} \to \pos{s_1}) \to (\src{F}_{s_1}~A \leadsto \src{G}_{s_2}~A) \\
\src{replicate} &:& (s : S) \to A \leadsto \src{F}_s~A \\
\src{tp} &:& \src{F}_{s_1}~(\src{F}_{s_2}~A) \leadsto \src{F}_{s_2}~(\src{F}_{s_1}~A) \\
\src{filter} &:& (\pos{s} \to \mathbb{B}) \to (A \times \src{F}_s~A \leadsto \src{F}_s~A) \\
\end{array}\]
\caption{Generic operations.}
\label{fig:genops}
\end{minipage}%
\begin{minipage}{.45\textwidth}
\[\begin{array}{lll}
\beta &:=& \mathbb{R} \\
(\cdot \oplus \cdot) &:=& (\cdot +_\mathbb{R} \cdot) \\
(\cdot \ominus \cdot) &:=& (\cdot -_\mathbb{R} \cdot) \\
\mathsf{S} &:=& \mathbb{N} \\
\pos{n} &:=& \{0, \dots, n-1\} \\
\den{\ext{relu}}~x &:=& \mathsf{if}~x<0~\mathsf{then}~0~\mathsf{else}~x\\
\den{\ext{*}}~(x, y) &:=& x \cdot y \\
\den{\ext{sum}}~x &:=& \Sigma_i (x~i)\\
\end{array}\]
\caption{Linear algebra instantiation (without incrementalization).}
\label{fig:linalginst}
\end{minipage}
\end{figure}

\Lang{} includes a number of generic operations (see Figure~\ref{fig:genops} for a subset, and the next section for the rest) that we deem useful for incremental functional programming.
We do not claim completeness, however, in Section~\ref{sec:app} we illustrate that they are expressive enough to implement interesting domain-specific operations for linear algebra, relational algebra and conflict-free replicated datatypes.
These generic operations are always available, no matter which data structure is used.
Note that $A \leadsto B$ denotes an incremental function in \Lang{}, while $A \to B$ denotes an arbitrary function in the meta language (in our case, Lean).
Further, $\den{A}$ is the (meta language) denotation of the object-level type $A$,
and $\src{F}_s~A$ is a container of shape $s$ with elements from $A$. In the case of linear
algebra, this means $\src{F}_s~A$ is an array of length $s$.

We provide second-order functions $\src{map}$, $\src{map2}$, and $\src{zip}$ with the usual semantics.
The function $\src{dup}$ duplicates a value and $\src{cst}~x$ is a constant function.
We allow single-element reads and writes via $\src{get}$ and $\src{set}$.
The operation $\src{reshape}~r$ applies a shape transformation by using function $r$ to transform indices. Its output can actually be of a different container type (here written $\mathsf{G}$) than its input.
For instance, we can express reversal of a list of length $n$ by defining $r$ to map index $i$ to $n-i-1$.
Next, $\src{replicate}$ produces a data structure where every entry holds the same value.
The function $\src{tp}$ transposes a nested data structure by swapping the outer and inner type constructors. 
Finally, $\src{filter}$ takes a boolean-valued function $p$ on indices and sets an entry to a default value if its index is mapped to false by $p$.
(In the next section, we describe the operators' denotational semantics in detail, cf.~Figure~\ref{fig:tmsem}.)
\Lang{} also offers operations for composing programs and for handling products and sums.
We omit them for now and discuss them later in Section~\ref{sec:fw}.

We define these operations in terms of the generic container representation.
This facilitates extensibility and reuse in \Lang{} since core operations (and their incrementalizations) are automatically available as soon as we define a data structure as a container.
We show that many data types can be encoded as containers, including matrices, relations, and the state vectors of CRDTs.
Hence, when a user wants to incrementalize a new application domain, they immediately gain access to a number of already incrementalized operations -- once they show that the underlying data type is a container.
The \Lang{} user only needs to provide incrementalizations for basic domain-specific operations (e.g., multiplication) and can then use the generic operations to construct more complex incremental programs.
This reduces the user's implementation effort and allows for very concise definitions of domain-specific instances, as illustrated by our case studies in Section~\ref{sec:app}.

\paragraph{Instantiating \Lang{} for Linear Algebra}\label{sec:laExtension}

\Lang{} supports modular domain-specific instantiations of the core calculus, facilitating adding support for specific application domains. As long as the domain-specific operations are given correct incrementalizations, all \Lang{} programs using them are guaranteed to have correct incrementalizations. Moreover, \Lang{}'s incrementalization combinators support the developers of domain-specific operations in providing correct incrementalizations.
Below, we introduce the instantiation mechanism informally through the example of linear algebra, as summarized in Figure~\ref{fig:linalginst}. This means defining a base type, a data structure, and a set of operations.

First, we define the \textit{base type} $\beta$ that the other types are built on top of.
The base type for our linear algebra instantiation is defined by~$\beta := \mathbb{R}$.
We also need to define the set of changes $\beta'$ and the update operation $\oplus$ applying a change to a value. In this case, $\beta' := \mathbb{R}$ and the update operation is addition.
We also define a subtraction operation that computes the difference between two values.
The above data determines a change structure~\cite{Cai14}, as explained further in Section~\ref{sec:fw}.
Next, we define that this instantiation's domain data structure are arrays.
To reuse the data-type-generic operations from the core calculus (Figure~\ref{fig:genops}), this definition is given by describing arrays as a container.
For this we define the shape type of array as $S := \mathbb{N}$, corresponding to the length of the array.
As mentioned above, we define the position type as $\pos{n} := \{0, \dots, n-1\}$, meaning that an array of length $n$ has valid indices from $0$ to $n-1$.
Defining arrays this way implies that a vector of length $n$ is represented as $\mathsf{F}_n~\beta$ and an $n\times m$-matrix is represented as $\mathsf{F}_n~(\mathsf{F}_m~\beta)$.

We only add three primitives for our linear algebra instantiation: 
The rectified linear unit ($\ext{relu} : \beta \to \beta$),
scalar multiplication ($\ext{*} : \beta \times \beta \to \beta$),  
and vector summation ($\ext{sum} : \mathsf{F}_s~\beta \to \beta$).
In Figure~\ref{fig:linalginst}, the semantics of each operation $o$ is written
$\den{o}$ (e.g. $\den{\ext{relu}}$).
In combination with the generic operations, these three domain-specific primitives are sufficient to define many more complex linear algebra operations.
For example, using our generic operations and the scalar multiplication,
we can construct matrix-vector multiplication as a derived operation.
In mathematical notation, matrix-vector multiplication is defined as follows:
\[(\mathsf{mvmul}~M~v)[i] = \textstyle \sum_j (M[i, j] \cdot v[j])\]
Operationally, we match up each element $x$ of $v$ with a column
of $M$, multiplying $x$ with every entry in the column and then summing up the results.
In \Lang{}, this is written as follows (where we map over two arrays at once using $\src{map2}~f$):
\[ \der{mvmul} := \src{(id \times replicate )~;~\src{map2}~(\src{map2}}~\ext{*} \src{)~;~map}~\ext{sum}. \]
The program $\der{mvmul}$ takes a pair $(M, v)$ as input and computes $M~v$. It does this in three steps:
First $\src{replicate}$ is applied to $v$, turning it into a matrix $v'$ by replicating every element.
Then, element-wise multiplication is performed between $M$ and $v'$. On the result of that, we apply row-wise summation.

Our $\der{mvmul}$ is a derived operation that is defined using both the core calculus and operations from the linear-algebra-specific instantiation. Namely, the operations $\ext{mul}$
and $\ext{*}$ are linear algebra primitives (defined in Figure~\ref{fig:linalginst}). Most of the constructs used to define $\der{mvmul}$ come from the core calculus, including the identity function~$\src{id}$, the function composition operator $\src{(}f~\src{;}~g\src{)}$, the parallel operator $\src{(}f~\src{\times}~g\src{)}$ which returns $(f~x, g~y)$ for input $(x,y)$, and the second-order operator $\src{map2}$.

For illustration, compare the mathematical description of a dense layer operating on vectors (left), with an implementation in \Lang{} (right):
\[
\begin{array}{c}
M : \mathsf{F}_m~(\mathsf{F}_n~\beta) ~\qquad\qquad~ b : \mathsf{F}_m~\beta ~\qquad\qquad~ \dense{} : \mathsf{F}_n~\beta \to \mathsf{F}_m~\beta \\[1em]
\dense{}~x := \mathsf{map}(\mathsf{relu}, M \cdot x + b) \qquad
  \begin{array}[t]{ll}
  \der{dense} :=& \pair{\src{cst}~M}{\src{id}} ; \der{mvmul} ; \\
  & \pair{\src{cst}~b}{\src{id}} \src{; map2~+ ;} \\
  & \src{map}~\ext{relu}
  \end{array}
\end{array}
\]

The function $\der{dense}$ uses our derived $\der{mvmul}$, the linear algebra primitive
$\ext{relu}$, the second-order operator $\src{map}$, and the constant function combinator $\src{cst}$ (we also make use of syntactic sugar: $\pair{f}{g} := \src{(dup;}~f~\src{\times}~g\src{)}$, which applies two functions to the same value by duplicating it first).

Informally, \dense{} is implemented as the composition of pushing the matrix $M$, multiplying $M$ with the argument of the function, adding $p$ to the result of the multiplication (using $\src{map2}$), and applying $\ext{relu}$ on each element.

\paragraph{Cached Incrementalization for the Linear Algebra Instantiation}
The final step of the linear algebra instantiation consists of defining the cached incrementalizations for the domain-specific operations.
The changes that we consider are element-wise modifications, where we add real numbers (that may be zero or negative) to each entry in the value. So for a vector, a change is another vector that is added element-wise.

To incrementalize $\ext{relu}$ and $\ext{*}$, we need to cache both input values.
In both cases, we compute the derivative by subtracting the old output from the new output.
For example, the derivative for $\ext{relu}$ is
\[\lambda (x', x).~(\den{\ext{relu}}~(x \oplus x') \ominus \den{\ext{relu}}~x,~x \oplus x'),\]
where $x'$ is the change, $x$ is the old input, $x + x'$ is the new input, and $\den{\ext{relu}}~(x + x')$ is the new output. The derivative takes the change and the old input. It returns the difference between the new output and the old output.

Note that for a scalar operation such as ReLU or multiplication by itself, incrementalization does not provide performance benefits.
On the contrary, it increases memory requirements due to caching. However, defining their incrementalizations, allows us to apply incrementalization \textit{compositionally} by combining the caches of subexpressions. For example, when applying $\src{map}$ to $\ext{*}$, we achieve element-wise vector multiplication with a vector-shaped cache. In this case, the cached incrementalization of multiplication ultimately enables an asymptotic speed-up: We only update one element of the output vector if one input element changes, whereas reevaluation would have to traverse an entire vector. More concretely, the incremental version $\src{map}$ takes a change and maps the incremental version of $\ext{*}$ over it. As the change only consists of one element, the incremental version of $\ext{*}$ is only called once.

We spelled out the incrementalization of ReLU explicitly, but we can also give a more concise definition using \textit{incrementalization combinators}, which we introduce in more detail in Section~\ref{sec:incrComb}. 
Intuitively, the incrementalization combinators reify common patterns in incrementalizations and can thereby simplify the implementation. For a function $f$ (that may need to fulfill some property, depending on the combinator), an incrementalization combinator returns a tuple $(\ctype{f}, \init{f}, \deriv{f})$ containing the cache type, initialization, and derivative.
When an incrementalization combinator's preconditions are fulfilled, it is formally proven to output a correct incrementalization.

As mentioned above, the incrementalizations of $\ext{relu}$ and $\ext{*}$ are both defined by caching their inputs and by computing the derivative using subtraction of the old output from the new one. We capture this pattern in the $\textsc{Triv}$ combinator, which creates a ``trivial'' incrementalization that performs reevaluation.
Another incrementalization combinator is $\textsc{Lin}$, which captures the pattern of incrementalizing linear functions, as a linear function is its own derivative.
An example of a linear function is the $\ext{sum}$ operator.
Accordingly, the incrementalizations can be defined as in terms of the normal evaluation semantics of the operators (we write $\den{f}'$ for the tuple $(f_C, f_i, f_d)$):
\[
\den{\ext{relu}}' := \textsc{Triv}~\den{\ext{relu}} \qquad\qquad
\den{\ext{*}}'    := \textsc{Triv}~\den{\ext{*}} \qquad\qquad
\den{\ext{sum}}'  := \textsc{Lin}~\den{\ext{sum}}
\]

%% file: framework.tex
In this section, we formally describe our approach.
First, we describe the syntax and denotational semantics of our calculus \Lang{}.
We also give the interface used for instantiating \Lang{} to new data structures and operations.
Adding new operations requires us to also define their incrementalization.
We propose a small set of combinators that simplify the implementation of correct-by-construction incrementalizations.
We then explain how our approach allows us to automatically and soundly incrementalize computations by defining an incrementalization transformation on \Lang{} terms. 

\subsection{Interface}
\label{sec:if}

As already mentioned, \Lang{} is parametrized over the choice of container and base type. 

Figure~\ref{fig:inter} presents \Lang{}'s parametric interface.
The figure's top half shows $n$-ary containers, which represent type constructors with $n$ arguments.
In this context, $S$~denotes the container's shape type while $\pos{\cdot}$ takes a shape
and gives an index type for each of the $n$ type arguments.
To ensure our proof is constructive, we require shapes to have decidable equality
(our proofs do depend on propositional and functional extensionality).

To aid comprehension we make two simplifications in the rest of the paper.
First, our mechanization supports multiple containers and base types.
Having multiple containers does allow to model, for example, nested different containers,
such as a relation whose elements contain arrays like in a vector database.  
In the paper, we will only show the case of a single container and base type.
Second, our mechanization supports containers with multiple type arguments ($n$-ary containers).
In the paper, we will only present containers with a single type argument (unary containers),
as all operations except $\src{tp}$\footnote{
transposition only makes sense on unary containers.
} generalize in a non-surprising fashion from the unary to the $n$-ary case,
we limit the presentation in the paper to the unary case.
Details on multiple containers, multiple base types and $n$-ary containers,
can be found in the file \texttt{Correctness.lean} in the artifact.

We require the base type to form a \emph{change structure}~\cite{Cai14}, i.e., a structure combining a type of values, a type of changes, and operations applying changes to values.
We denote the base value type by~$\beta$ and its
corresponding type of changes by~$\beta'$.
We require an update operation~$\oplus$ and corresponding difference operation~$\ominus$, such that $x \oplus (y \ominus x) = x$.
The update operation allows us to formalize the notion of change.
The difference operation is used by the \textsc{Triv} combinator to compute output differences and for the derivative of operations whose output is partly or wholly constant (as $x \ominus x$ is the \emph{nil change} for $x$, i.e. we have $x \oplus (x \ominus x) = x$ by the just-mentioned law).
Our change structures differ from the original version presented by \citet{Cai14} in that our change type is fixed rather than dependent on the value.
In this paper, our examples and figures assume that are there is only one base type $\beta$. This is just for simplicity of the presentation---our mechanization allows for an arbitrary number of base types.

\subsection{Syntax and Semantics}

\begin{figure}
\begin{minipage}{.5\textwidth}
\[\begin{array}{lll}
\multicolumn{3}{l}{\textsc{Container}}\\
  S &:& \mathsf{Type} \\
  \pos{\cdot} &:& S \to \mathsf{Type} \\
  \mathsf{deceq} &:& (s : S) \to \mathsf{DecidableEq}~s
  \\[1em]
\multicolumn{3}{l}{\textsc{Change Structure}}\\
  \beta &:& \mathsf{Type} \\
  \beta' &:& \mathsf{Type} \\
  (\cdot \oplus \cdot) &:& \beta \to \beta' \to \beta \\
  (\cdot \ominus \cdot) &:& \beta \to \beta \to \beta' \\
  \mathsf{complete} &:& x \oplus (y \ominus x) = x
\end{array}\]
\caption{Container and Change Structure Interface.}
\label{fig:inter}
\end{minipage}%
\begin{minipage}{.5\textwidth}\centering
$s : \mathsf{S} \qquad r : \pos{s_2} \to \pos{s_1} \qquad c : \beta$ \\[.5em] $i : \pos{s} \qquad
\ext{o} : \mathsf{Op}~A~B \qquad p : \pos{s} \to \mathbb{B}$\\
  \[ \begin{array}{rll}
  \multicolumn{2}{l}{\textsc{Types}} & \\
   A, B  ::= & \src{b} \mid \src{F}_s~A \mid A~\times~B \mid A~+~B \\[1em]
  \multicolumn{2}{l}{\textsc{Expressions}} & \\
  e, f ::= & e_1~\src{;}~e_2 \mid e_1~\src{\times}~e_2 \mid \src{id} \mid \src{dup} \mid \src{fst} \mid \src{snd} \\
  \mid & \src{map}~e  \mid \src{zip} \mid \src{get}~i \mid \src{set}~i \mid 
  \src{reshape}~r \\
  \mid & \src{replicate}~s \mid \src{tp} \mid \src{filter}~p \mid \src{fuse} \mid \src{distr} \\
  \mid & \src{inl} \mid \src{inr} \mid e_1~\src{||}~e_2 \mid \src{cst}~c \mid \src{+} \mid \src{op}~\ext{o}
\end{array} \]
\caption{Syntax.}
\label{fig:syntax}
\end{minipage}
\end{figure}

\begin{figure}
\begin{mathpar}
\inferrule[]{\vdash e_1 : A \leadsto B \\ \vdash e_2 : B \leadsto C}{\vdash e_1~\src{;}~e_2 : A \leadsto C}

\inferrule[]{\vdash e_1 : A \leadsto B \\ \vdash e_2 : C \leadsto D}{\vdash e_1~\src{\times}~e_2 : A\times C \leadsto B\times D}

\inferrule[]{ \ext{o} : \mathsf{Op}~A~B}{\vdash \src{op}~\ext{o} : A \leadsto B}

\inferrule[]{\vdash e : A \leadsto B}{\vdash \src{map}~e : \src{F}_s~A \leadsto \src{F}_s~B}

\inferrule[]{ \vdash e_1 : A \leadsto B \\ \vdash e_2 : C \leadsto D }{ \vdash e_1~\src{||}~e_2 : A + C \leadsto B + D }

\vdash \src{id} : A \leadsto A \quad
\vdash \src{fst} : A \times B \leadsto A \quad \vdash \src{snd} : A \times B \leadsto B \quad
 \vdash \src{+} : A \times A \to A^\dag \\ \vdash \src{distr} : A \times (B + C) \leadsto
A \times B + A \times C \quad
\vdash \src{fuse} : A + A \leadsto A \quad \vdash \src{inl} : A \leadsto A + B
\quad \vdash \src{inr} : B \leadsto A + B
\end{mathpar}
\caption{Types of core constructs (we omit those already shown in Figure~\ref{fig:genops}). $^\dag$$\src{+}$ requires that values and changes are the same, i.e. $\beta = \beta'$, as its semantics uses $\oplus$, which would be ill-typed otherwise.}
\label{fig:typ}
\end{figure}

\begin{figure}
\begin{mathpar}
\inferrule[]{f : \den{A} \to \den{B} \\ C : \mathsf{Type} \\ i : A \to B \times C \\ d : A' \to C \to B' \times C }{(f, C, i, d) : \mathsf{Op}~A~B}
\end{mathpar}

\caption{User-defined operations.}
\label{fig:ext}
\end{figure}

In this section we describe \Lang{}'s syntax, type system, and denotational semantics.
Figure~\ref{fig:syntax} presents the syntax of our core calculus.
We denote shapes by $s:S$, shape transformations by $r$, constants by $c:\beta$, positions by $i$, operations by $\ext{o}$ and predicates by $p$.
A type is either the base type $\src{b}$, a data structure $\src{F}_s~A$ of shape $s$, a product type $A ~\times~ B$, or a sum type $A~+~B$.
As the shape is part of the type, we statically know the shapes of all values computed by any \Lang{} program.
Static shape typing also ensures totality, since operations like $\mathsf{get}$ cannot go out of bounds.
This simplifies~the~denotational~semantics.

\Lang{} supports the following expressions $e$:
(i)~Sequential composition of two expressions $e_1~\src{;}~e_2$,
(ii)~parallel composition of two operations $e_1~\src{\times}~e_2$,
(iii)~literals~$\src{cst}~c$.
(iv)~identity~$\src{id}$,
(v)~duplication~$\src{dup}$,
(vi)~tuple projection~$\src{fst}$ and~$\src{snd}$,
and (vii)~addition~$\src{+}$.
The literals must be valid values according to the type denotation defined below.

In addition to tuples, \Lang{} also supports sum types. The operators for products and sums are based
on the structure of distributive categories~\cite{Elliott17}. In particular, sums can be constructed with
$\src{inl}, \src{inr}$ and manipulated with $e_1~\src{||}~e_2, \src{fuse}$, and $\src{distr}$.

\Lang{}'s expressions also include container operations\footnote{We omit the previously shown $\src{map2}$ as it can be defined by composing $\src{map}$ and $\src{zip}$}:
(i)~$\src{get}$ retrieves and $\src{set}$ sets the~container element at index $i : \pos{s}$. 
(ii)~$\src{replicate}$ creates a data structure with all elements set~to~the~same value.
(iii)~$\src{filter}$ takes a predicate on indices $p$ and zeroes out all elements at indices $i$ where $p~i$ is false.
(iv)~$\src{map}~e$ maps a function over a container.
(v)~$\src{zip}$ combines two equally shaped containers, pairing up their elements.
(vi)~$\src{tp}$ transposes nested data structures.
(vii)~$\src{op}$ allows calling a user-provided operation $\ext{o}$.
(viii)~$\src{reshape}$ changes a container's shape by rearranging its contents via a shape transformation~$r : \pos{s_2} \to \pos{s_1}$.
The latter maps indices in $s_2$ to $s_1$.

We present the type system for \Lang{} in Figure~\ref{fig:typ}.
The typing judgment has the form $\vdash e : A \leadsto B$, meaning that $e$ transforms values of type $A$ into values of type $B$.
Figure~\ref{fig:ext} formalizes our approach to specify the additional operations that ship with a domain-specific instantiation.
User defined functions~$\ext{o}$ are packed into the $\mathsf{Op}$ type.
This type associates functions with their cached incrementalization.
The latter include a cache $C$, an initialization $i$, and a derivative $d$.
(We require them to comply with laws ensuring soundness of the incrementalization, cf.~Section~\ref{sec:corrst}.)

Figure~\ref{fig:tmsem} presents \Lang{}'s denotational semantics.
The semantics is implicitly parametrized over the container and change structure interfaces described above.
Types are given both a value semantics $\den{\cdot}$ and a change semantics $\den{\cdot}'$, which only differ in the case for the base type.
In both cases, we interpret container types as functions from indices to values.
The change semantics for sums allows to either (i) change a value while staying on the same ``side'' of the sum (i.e., going from $\iota_1~x$ to $\iota_1~(x+x')$ or from $\iota_2~x$ to $\iota_2~(x+x')$) (ii) set a new value.
For space reasons, it is shown in Figure~\ref{fig:sumcs} in Appendix~\ref{sec:sums}.

As to the term semantics, we describe the results of data transformations by stating the value at each index. 
For instance, $\den{\src{zip}}~(x, y)$ yields a container with entries~$(x~i, y~i)$ at each index~$i$. As another example, $\den{\src{reshape}~r}~x$ outputs a data structure which has at index $i$ the value taken from index~$r~i$ in~$x$.

\begin{figure}
\begin{minipage}{.24\textwidth}\centering
$\inferrule{\vdash e : A \leadsto B}{\den{e} : \den{A} \to \den{B}}$
\end{minipage}%
\begin{minipage}{.38\textwidth}\centering
$\begin{array}{lll}
\den{\src{b}} &:=& \beta \\
\den{\src{F}_s~A} &:=& \pos{s} \to \den{A} \\
\den{A~\times~B} &:=& \den{A} \times \den{B}
\end{array}$
\end{minipage}%
\begin{minipage}{.38\textwidth}\centering
$\begin{array}{lll}
\den{\src{b}}' &:=& \beta' \\
\den{\src{F}_s~A}' &:=& \pos{s} \to \den{A}' \\
\den{A~\times~B}' &:=& \den{A}' \times \den{B}'
\end{array}$
\end{minipage}

\vspace{1em}

\begin{minipage}{.5\textwidth}
$\begin{array}{lll}
\den{\src{id}}~x &:=& x \\
\den{\src{cst}~c}~x &:=& c \\
\den{e_1~\src{;}~e_2}~x &:=& \den{e_2}~(\den{e_1}~x) \\
\den{e_1~\src{\times}~e_2}~(x, y) &:=& (\den{e_1}~x, \den{e_2}~y) \\
\den{\src{dup}}~x &:=& (x, x) \\
\den{\src{+}}~(x,y) &:=& x \oplus y \\
\den{\src{fst}}~(x,y) &:=& x \\
\den{\src{snd}}~(x,y) &:=& y \\
\den{\src{inl}}~x &:=& \iota_1~x \\
\den{\src{inr}}~x &:=& \iota_2~x \\
\den{\src{distr}}~(x, \iota_1~y) &:=& \iota_1~(x, y) \\
\den{\src{distr}}~(x, \iota_2~y) &:=& \iota_2~(x, y) \\
\end{array}$
\end{minipage}%
\begin{minipage}{.5\textwidth}
$\begin{array}{lll}
\den{\src{replicate}~s}~x &:=& \lambda i.~x \\
\den{\src{map}~e}~x &:=& \lambda i.~\den{e}~(x~i) \\
\den{\src{zip}}~(x, y) &:=& \lambda i.~(x~i, y~i) \\
\den{\src{reshape}~r}~x &:=& \lambda i.~x~(r~i) \\
\den{\src{get}~i}~x &:=& x~i \\
\den{\src{set}~i}~(x, a) &:=& \lambda j.~\mathsf{if}~i=j~\mathsf{then}~x~\mathsf{else}~a~j \\
\den{\src{tp}}~x &:=& \lambda i.~\lambda j.~x~j~i \\
\den{\src{filter}~p}~(x, a) &:=& \lambda i.~\mathsf{if}~p~i~\mathsf{then}~a~i~\mathsf{else}~x \\
\den{\src{fuse}}~(\iota_1~x) &:=& x \\
\den{\src{fuse}}~(\iota_2~x) &:=& x \\
\den{e_1~\src{||}~e_2}~(\iota_1~x) &:=& \den{e_1}~x \\
\den{e_1~\src{||}~e_2}~(\iota_2~x) &:=& \den{e_2}~x \\
\end{array}$
\end{minipage}
\caption{Semantics of types and terms.}
\label{fig:tmsem}
\end{figure}

\subsection{Incrementalization}
\label{sec:tf}

In this section, we describe the incrementalization function.
For every type of values $A$, we refer to its corresponding type of changes as $A'$.
We also denote values by $x$ and changes by $x'$.
The incrementalization of a function $f : A \to B$ consists of:
\[\begin{array}{lll}
\ctype{\den{f}} &:& \mathsf{Type} \\
\init{\den{f}} &:& A \to B \times \ctype{\den{f}} \\
\deriv{\den{f}} &:& A' \to \ctype{\den{f}} \to B' \times \ctype{\den{f}} \\[.5em]
\den{f}' &:=& (\ctype{\den{f}},~ \init{\den{f}},~ \deriv{\den{f}}) \\
\end{array}\]
$\ctype{\den{f}}$ is the cache type. 
The initialization function $\init{\den{f}}$ takes a value $x$ and returns $(f~x, c)$ where $c$ is the initial cache.
The derivative $\deriv{\den{f}}$ takes an input change $x'$ and the current cache. It returns an output change and the new state of the cache. We write $\den{f}'$ as a~shorthand~for~$(\ctype{\den{f}},~ \init{\den{f}},~ \deriv{\den{f}})$.

We present the incrementalization in Figure \ref{fig:incrRules}.
The basic operations are defined using the \textsc{Self} and \textsc{Add} combinators, which we explain in Section~\ref{sec:incrComb}.
Incrementalizing a sequential composition~$\src{(} f \src{;} g \src{)}$ combines the caches of the incrementalized subexpressions. 
To evaluate the composition's derivative, we split the cache into parts for $f$ and $g$, evaluate their respective derivatives, and then recombine the updated caches.
The same principle applies to parallel composition $\src{(} f \src{\times} g \src{)}$.
When we apply $\src{map}_s~f$ to a data structure, we essentially perform repeated parallel composition. Hence, the cache type is $\pos{s} \to \ctype{\den{f}}$, with one copy of the cache type of $f$ for each position in the data structure. 
Similar to parallel composition, each call to $f$ receives its own version of the cache.
The incrementalizations for $\src{fuse}, \src{distr}, \src{||}$ are rather lengthy,
and hence delegated to Appendix~\ref{sec:sums}.

\input{figures/incrRules}

\subsection{Self-maintainability}

For some functions, the output change only depends on the input change
but not on the input itself. Formally the derivative $d$ has the following property:
\[\forall x~y~x'.~d~x~x' = d~y~x'. \]
Such functions are called \textit{self-maintainable}~\cite{Cai14}.
A self-maintainable function can be incrementalized without storing its intermediate results, 
i.e., the cache type can be $\mathsf{Unit}$. 
This is because caching is only necessary if we require additional information besides the input change.
Our generic operations are self-maintainable and the combinators of our core calculus preserve self-maintainability.

In incrementalization, this property can be exploited for efficiency,
as self-maintainable functions do not require a cache and only operate on changes, rather than whole values.
\citet{Cai14} realize the efficiency benefit of self-maintainability using lazy evaluation, but mention dead code elimination as an alternative.
In our approach, programs built out of self-maintainable operations automatically
yield incrementalizations with an empty cache and therefore only operate on changes.

As an example, consider the $\src{zip}$ function. A change to one of the arguments of $\src{zip}$ is passed on as a change of the output, without using the current value. Concretely, we have $\den{\src{zip}}' (x \oplus x', y \oplus y') := \den{\src{zip}}~(x', y')$.
Hence, computing $\den{\src{zip}}'$ does not use $x$ or $y$, but purely inspects the changes.

\subsection{Incrementalization Combinators}
\label{sec:incrComb}

\input{figures/incrComb}

When instantiating \Lang{} to a specific domain, the user has to provide incrementalizations for newly added primitive operations. 
To aid in the implementation of such incrementalizations, we provide a set of \text{incrementalization combinators}, presented in Figure~\ref{fig:incrComb}.
Each combinator has some inputs and preconditions and returns a tuple $(C, i, d)$ such that $d$ is the derivative of a function $f$, with $C$, $i$, and $d$ having the right types as declared in Section~\ref{sec:tf}.

For any function $f$, we can use the \textsc{Triv} combinator to create a \textit{trivial} incrementalization, i.e., one that just recomputes values upon changes.
Technically, this construction meets the formal requirements necessary to consider it an incrementalization, yet it does not provide any speed-up.
However, such incrementalizations can still be useful: 
Consider scalar operations such as real multiplication.
Since asymptotic speed-ups via incrementalization are impossible here, we may as well use the trivial incrementalization.
The \textsc{Triv} combinator relies on the difference function~$\ominus$.
For any values $a, b$ it computes a change $a \ominus b$ that transforms $b$ into $a$. 
For example, if $A$ is the set of real numbers and the update operation is addition, then~$\ominus$ is subtraction.
For scalar multiplication, we write the~trivial~incrementalization as $\textsc{Triv}~\den{\ext{*}}$. 
The induced derivative is then simply $\lambda ((x_1', x_2'), (x_1, x_2)). (y \ominus (x_1 \cdot x_2), x \oplus x')$, where $y := (x_1 \oplus x_1') \cdot (x_2 \oplus x_2')$ is the new output.
The \textsc{Triv} combinator has the drawback of evaluating $f$ twice.
We therefore also implement the \textsc{Triv2} combinator, which only evaluates $f$ once, by storing both $f$'s input \emph{and output}.
This allows trading off memory for speed.

The combinator \textsc{Self} simplifies incrementalization for the special case of self-maintainable functions. Recall that in the context of incrementalization, a function is self-maintainable if an output change only depends on the input change, but not on the current value.
In other words, it has a constant derivative. More precisely, for a self-maintainable function $f : A \to B$, there exists a function $d : A' \to B'$ such that $f~(x \oplus x') = f~x \oplus d~x'$.
For instance, the function $\src{(cst}~c\src{)}$ is self-maintainable, with derivative 0.
This essentially holds because $\src{cst}$ has a constant output and changes are discarded.
This makes the derivative~$\den{\src{cst~}0}$ and we can define its incrementalization as
$\textsc{Self}~\den{\src{cst~}c}~\den{\src{cst~}0}$.

The \textsc{Self} combinator constructs an incrementalization that does not require a cache, which is indicated by defining the cache type to be $\mathsf{Unit}$.
(We write $\star$ for the sole element of the unit type.)
In \Lang{}, a surprisingly large number of operations is implemented using \textsc{Self}, cf.~Figure~\ref{fig:incrRules}.
For example, it is straightforward to see that $\src{dup}$ is self-maintainable, as we have
$\den{\src{dup}}~(x \oplus x') = (x \oplus x', x \oplus x') = (x, x) \oplus (x', x') = \den{\src{dup}}~x \oplus \den{\src{dup}}~x'$.

The remaining combinators only apply in the special case where values and changes form one type, i.e. $A = A'$. To show this for all types, it is only necessary to show it for the base type.

Incrementalizing \textit{linear} functions, which are their own derivative, is even simpler.
Essentially, linear functions are self-maintainable functions where $f = d$.
Note that this is only well-defined for functions $f : A \to B$ where $A = A'$ and $B = B'$ (this follows if $\beta = \beta'$ i.e. the base value and change types are the same).

We can also incrementalize \textit{bilinear} functions, which are binary functions that are linear in their first and second arguments. In contrast to linear functions, incrementalizing bilinear functions does require a cache~\cite{Budiu23}.
An example of a bilinear function is the relational cross-product.

Finally, if the update operation $(\cdot \oplus \cdot)$ is associative and commutative, it can also be used as an operation such that it is its own derivative.
This is because the soundness property for incrementalization of $(\cdot \oplus \cdot)$ is
\[(x~\oplus~x')~\oplus~(y~\oplus~y') = (x~\oplus~y)~\oplus~(x'~\oplus~y'),\]
which clearly follows from associativity and commutativity.

\subsection{Growing Data under Static Shapes}

Our core calculus enforces static container shapes by design:
a container's shape is fixed at the type level and cannot change dynamically during execution.
This ensures the totality of key operations such as $\src{set}, \src{get}, \src{tp}$.
Shape changing operations can be implemented, as long as the shape is still statically fixed.
For example, we can implement $\ext{append}$ as function going from shape $n$ to shape $n+1$: %
\[ \ext{append} : A \times \src{F}_n~A \leadsto \src{F}_{n+1}~A := (\src{id}~\src{\times}~\src{reshape}~(\lambda i.~i-1))~\src{;}~\src{set}~0 \]

Limiting the change structures to maintain the current shape also enables operations $\src{get}$ to be self-incrementalizable. 
Despite that, it is also possible to implement growing data structures, as we show in our relational algebra and tree examples in Section~\ref{sec:app}.
Where we included the dimension of an tensor in the shape in our linear algebra case study, e.g., the length of an array is tracked statically,
our relational algebra case study only tracks the relation's \textit{schema} statically, e.g., the number of items in the relation can grow dynamically.
Further, in the tree case study, we simply treat the shape as the unit type, which effectively means treating all shapes dynamically.

%% file: figures/incrRules.tex
\begin{figure}

\raggedright
$
\begin{array}[t]{lllll}
  \den{f}' &:=& \textsc{Self}^\diamond~\den{f} 
  &&f \in \{\src{id}, \src{dup}, \src{fst}, \src{snd}, \src{zip}, \src{tp},
          \src{get}~i, \src{set}~i, \src{reshape}~r, \src{replicate}~s\\
  &&&&\phantom{f \in \{}\src{filter}~p, \src{cst}~c, \src{inl}, \src{inr}\}
\end{array}
$
\vspace{0.5em}
\\
$\begin{array}[t]{lll}
\ctype{\den{\src{map}_s~f}} &:=& \pos{s} \to \ctype{\den{f}} \\
\init{\den{\src{map}_s~f}}~x &:=& \mathsf{let}~a=\lambda i.~\init{\den{f}}~(x~i); \\
&& (\lambda i. (a~i)_1,~\lambda i. (a~i)_2) \\
\deriv{\den{\src{map}_s~f}}~x'~c &:=& (\lambda i.~(\deriv{\den{f}}~(x'~i)~(c~i))_1, \\
&& \phantom{(}\lambda i.~(\deriv{\den{f}}~(x'~i)~(c~i))_2) \\
\end{array}\hfill\begin{array}[t]{lll}
\den{\src{cst}~c}' &:=& \textsc{Self}~\den{\src{cst}~c}~\den{\src{cst}~0} \\
\den{\src{+}}' &:=& \textsc{Add}^\dag \\
\den{\src{op}~\ext{o}}' &:=& \ext{o} \\
\end{array}\hfill$

\vspace{1em}

$\begin{array}{lllll}
\ctype{\den{f~\src{;}~g}}\ &:=&
  \ctype{\den{f}} \times \ctype{\den{g}} \\
\init{\den{f~\src{;}~g}}~x &:=&
  \mathsf{let}~(y,c_1)=\init{\den{f}}~x; &
  \mathsf{let}~(z,c_2)=\init{\den{g}}~y; &
  (z, (c_1, c_2)) \\
\deriv{\den{f~\src{;}~g}}~x'~c &:=&
  \mathsf{let}~(y', c_1')=\deriv{\den{f}}~x'~c_1; &
  \mathsf{let}~(z', c_2')=\deriv{\den{g}}~y'~c_2; &
  (z', (c_1', c_2')) \\
\ctype{\den{f~\src{\times}~g}} &:=&
  \ctype{\den{f}} \times \ctype{\den{g}} \\
\init{\den{f~\src{\times}~g}}~(x_1, x_2) &:=&
  \mathsf{let}~(y_1,c_1)=\init{\den{f}}~x_1; &
  \mathsf{let}~(y_2,c_2)=\init{\den{g}}~x_2; &
  ((y_1, y_2), (c_1, c_2)) \\
\deriv{\den{f~\src{\times}~g}}~(x_1', x_2')~c &:=&
  \mathsf{let}~(y_1', c_1')=\deriv{\den{f}}~x_1'~c_1; &
  \mathsf{let}~(y_2', c_2')=\deriv{\den{g}}~x_2'~c_2; &
  ((y_1', y_2'), (c_1', c_2')) \\
\end{array}\hfill$

\caption{Incrementalization on \Lang{} terms.
The \textsc{Self} combinator incrementalizes self-maintainable, i.e., cache-free operations. We write $\textsc{Self}^\diamond~x$ as shorthand for
$\textsc{Self}~x~x$. $^\dag$Incremental addition requires associativity and commutativity (see Section~\ref{sec:incrComb}).
}
\label{fig:incrRules}
\end{figure}

%% file: figures/incrComb.tex
\begin{figure}

\small
\begin{mathpar}
    \inferrule[]{f : A \to B}{
        \textsc{Triv}~f := (A,~\lambda x.(f~x, x),~\lambda (x', x).(f~(x \oplus x') \ominus f~x,~x \oplus x'))    
    }

    \inferrule[]{f : A \to B}{
        \textsc{Triv2}~f := (A \times B,~\lambda x.(f~x, (x, f~x)),~\lambda (x', (x, y_1)).
            \mathsf{let}~y_2 := f~(x \oplus x'); (y_2 - y_1, (x \oplus x', y_2)))    
    }

    \inferrule[]{f : A \to B \\ d : A' \to B' \\\\ f~(x \oplus x') = f~x \oplus d~x'}{
        \textsc{Self}~f := (\mathsf{Unit},~\lambda x.(f~x, \star),~\lambda (x', \star). (d~x', \star))
    }
	
    \inferrule[]{f : A \to B \\ A = A' \\ B = B' \\\\ f~(x \oplus x') = f~x \oplus f~x'}{
        \textsc{Lin}~f := (\mathsf{Unit},~\lambda x. (x, \star),~\lambda (x', \star). (f~x', \star))
    } 

    \inferrule[]{f : A \times B \to C \\ A = A' \\ B = B' \\ C = C' \\ f~(x \oplus x',~y) = f~(x, y) \oplus f~(x', y) \\ f~(x, y \oplus y') = f~(x, y) \oplus f~(x, y') \\ \forall (x~y : C).~x \oplus y = y \oplus x \\ \forall (x~y~z : C).~x \oplus (y \oplus z) = (x \oplus y) \oplus z}{
        \textsc{BiLin}~f := (A \times B,~\lambda (x, y).~(f~(x, y), x, y),~\lambda (x', y')~(x, y).~(f~(x', y') \oplus f~(x', y) \oplus f~(x, y'), x \oplus x', y \oplus y'))
    }

    \inferrule[]{A = A' \\ \forall (x~y : A).~x \oplus y = y \oplus x \\ \forall (x~y~z : A).~x \oplus (y \oplus z) = (x \oplus y) \oplus z}{
    \textsc{Add} := (\mathsf{Unit},~\lambda (x,y).(x \oplus y, \star),~\lambda (x', y', \star). (x' \oplus y', \star))
    }
\end{mathpar}    

\caption{Incrementalization combinators.}
\label{fig:incrComb}
\end{figure}

%% file: correctness.tex
In this section, we describe our Lean mechanization, including a proof that our incrementalization is \emph{\sound{}}, meaning that it computes the same result as full re-evaluation.
This captures the idea that incrementalization is an optimization that should avoid unnecessary recomputations, but otherwise preserve the original program behavior.

The section is structured into five subsections that systematically establish correctness.
First, we establish \emph{typability preservation} through intrinsic types (\S~\ref{sec:Correctness:IntrinsicTyping}).
Specifically, we demonstrate how our mechanization employs intrinsic types to guarantee typability across denotational semantics, incrementalization procedures, and combinator implementations.
Second, %
to formalize incrementalization, we state a \emph{formal definition of data type changes} as change structures (\S~\ref{sec:Correctness:Changes}). As noted in Section~\ref{sec:fw}, we assume that the base type forms a change structure.
We then prove that this ensures that all types (i.e. arbitrarily nested containers, products, and sums) form a change structure, so that a meaningful notion of change exists for any type.
Third, introduce the notion of a \emph{consistent incrementalization} to capture the requirements that the elements of the incrementalization tuple fulfill, and prove that \emph{all consistent incrementalizations are \sound{}}, i.e., produce identical results to full recomputation (\S~\ref{sec:corrst}).
Fourth, it is only left to show that our \emph{transformation and incrementalization combinators are consistent}, and hence \sound{} (\S~\ref{sec:incrsound}).
Finally, we state and prove under which conditions \Lang{}'s operations ensure that containers have \textit{finite support}, which is relevant for allowing finite representations of infinite containers like relations (\S~\ref{sec:finsupp}).

\subsection{Typability Preservation}\label{sec:Correctness:IntrinsicTyping}

We employ \textit{intrinsic typing}, meaning that the terms in
\lstinline{Tm} are parametrized over the input and output type.
The use of intrinsic typing means that the typing rules are directly integrated into the syntax declaration and hence only well-typed terms can be constructed.

So, a program that transforms a pair of values (from a base type \lstinline{b}) into a single value would have type
\lstinline{Tm (prod b b) b}.
Below in the left-hand listing, we sketch a small excerpt of the definition of \lstinline{Ty} and \lstinline{Tm} (corresponding to object types and expressions from Fig.~\ref{fig:syntax}), eliding some type parameters. \\
\begin{minipage}[t]{0.55\textwidth}
\begin{lstlisting}
inductive Ty : Type
  | base : Type | prod : Ty → Ty → Type | ...
inductive Tm : Ty → Ty → Type (u+1)
  | id : Tm A A | seq : Tm A B → Tm B C → Tm A C | ...
\end{lstlisting}
\end{minipage}
\begin{minipage}[t]{0.45\textwidth}
\begin{lstlisting}
def Tm.denote :
  Tm A B → A.toCS.V → B.toCS.V
def Tm.incr (e : Tm A B) :
  ConsIncr A.toCS B.toCS e.denote
\end{lstlisting}
\end{minipage}

\begin{lstlisting}
structure ConsIncr (A B : ChangeStructure) (f : A.V → B.V) : Type 1 where
  C : Type
  i : A.V → B.V × C
  d : A.D → C → B.D × C
  h₀ : ∀ x, (i x).1 = f x
  h₁ : ∀ x x', f (x + x') = f x + (d x' (i x).2).1
  h₂ : ∀ x x', (d x' (i x).2).2 = (i (x + x')).2
\end{lstlisting}

The right-hand listing above declares the denotational semantics and incrementalization. The declaration ensures that the denotation and incrementalization both have the correct type for the term \lstinline{e}.
Using intrinsic types, ensures that both \lstinline{Tm.denote} and \lstinline{Tm.incr} are well-defined total functions, which is an essential property to simplify further reasoning and proofs.

Note that \lstinline{toCS : Ty → ChangeStructure} is the denotational semantics for types, associating each type with a change structure (see more in Section~\ref{sec:Correctness:Changes}),
and \lstinline{ConsIncr} is a cached incrementalization of a function along with a consistency proof, which we discuss further in Section \ref{sec:corrst}.

\subsection{Containers and Change Structures}\label{sec:Correctness:Changes}

Our development is parametric over a container, a base type, and a set of domain-specific operations. We now describe the properties these three parameters need to satisfy to ensure correct incrementalization.
(i) As mentioned in Section~\ref{sec:if}, a container consists of a type of shapes and a function mapping each shape to a type of indices.
In the Lean formalization, we also require that each shape's index type has decidable equality. This is not a difficult requirement, as most containers indices will have decidable equality of indices, to allow to distinguish the positions of their elements. %
(ii) The base type forms a change structure containing the types of values and changes, an update, difference and a proof of completeness.
(iii) For adding further data-type-specific operations, the syntax type \lstinline{Tm} contains an operation \lstinline{op} with the type signature \lstinline{op : Op A B → Tm A B}.
An element of type \lstinline{Op} is an arbitrary incrementalized operation from \lstinline{A} to \lstinline{B}.
It is defined by giving a function \lstinline{f} from \lstinline{A} to \lstinline{B} together with a correct incrementalization of \lstinline{f}.

These three parameters are defined below.

\noindent\begin{minipage}[t]{.33\textwidth}
\begin{lstlisting}
structure Container where
  S : Type
  shape : S → Type
  [h: ∀ s, DecidableEq (shape s)]
\end{lstlisting}
\end{minipage}%
\begin{minipage}[t]{.36\textwidth}
\begin{lstlisting}
structure ChangeStructure where
  (A : Type)  (A' : Type)
  update : A → A' → A
  diff : A → A → A'
  complete : update x (diff y x) = y
\end{lstlisting}
\end{minipage}%
\begin{minipage}[t]{.3\textwidth}
\begin{lstlisting}
structure Op (A B : Ty) where
  f : A.toCS.V → B.toCS.V
  i : ConsIncr A.toCS B.toCS f
\end{lstlisting}
\end{minipage}

To talk about changes, we need to equip every object type with an algebraic structure which allows us to apply updates.
Hence, our first lemma is that if the given base types form a change structure,
our object types (Figure~\ref{fig:inter}, Figure~\ref{fig:tmsem}) will also form a change structure:

\begin{lemma}\quad\\
If the semantics of the base types $~\den{\src{b}}$ form a change structure,
then the semantics of all the object types $\den{A}$ (types constructed from base types, containers, products and sums, Fig.~\ref{fig:syntax}) also form a change structure.
\end{lemma}
\begin{proof}
By induction on $A$.
\end{proof}

Our Lean mechanization of the proof is constructive: We define a function \lstinline{toCS} that maps every type to a change structure. An element of type \lstinline{ChangeStructure} is a tuple consisting of $(A, A', \oplus, \ominus)$, as well as the $\mathsf{complete}$ law: $x \oplus (y \ominus x) = y$.

\subsection{Consistent Incrementalizations are \Sound{}}
\label{sec:corrst}

We now formally state what it means for a cached incrementalization to be correct. First, we state three \emph{consistency} laws that an incrementalization should fulfill. We then define a \emph{\soundness{}} property formalizing the idea that an incrementalization should produce the same results as re-evaluation when applied to a list of changes. We then show that consistency implies \soundness{}.

In the following, we assume $A$ and $B$ to be change structures.
For a function $f : A \to B$, we recall the functions produced by incrementalization and state the three laws they have to satisfy.

\begin{definition}[Incrementalization]\label{def:incrementalization}\quad\\
The tuple $(C: \mathsf{Type},~ i: A \to B \times C,~ d: A' \to C \to B' \times C)$ is a \emph{consistent incrementalization} of $f : A \to B$ if
\begin{align}
\forall x.    && i_1~x      &= f~x                    \tag{Law-1} \label{def:law1}\\
\forall x~x'. && f~(x \oplus x')     &= f~x \oplus d_1~x'~(i_2~x)  \tag{Law-2} \label{def:law2}\\
\forall x~x'. && i_2~(x \oplus x') &= d_2~x'~(i_2~x). \tag{Law-3} \label{def:law3}
\end{align}
\end{definition}
We denote the first component of the function result by $i_1$ and $d_1$, and the second component by $i_2$ and $d_2$.
\ref{def:law1} requires the first component of the initialization $i~x$ to be equal to the original function $f~x$. %
\ref{def:law2} requires the incrementalization $d_1$ to be a discrete derivative for $f$. This means that, if $d_1$ is given a cache for $x$ and a change $x'$, it will return the difference between $f~x$ and $f~(x+x')$.
\ref{def:law3} requires the incrementalization $d_2$ to also maintain the cache correctly. This means that, if we initialize the cache from an input $x$ and then apply an update $x'$, the new cache equals the one we would have gotten by initializing from $x + x'$. %

We now state the main property that should hold for our incrementalization, namely \soundness{}. Intuitively, an incrementalization is \emph{\sound{}} if it computes the same result as the original computation, given an arbitrary sequence of changes. 
More precisely, we define a function $\mathsf{iter}$ that applies the incrementalization on a list of input changes. We show that $\mathsf{iter}$'s result is equal to applying the function to the changed input.

The definition of $\mathsf{iter}$ is as follows (where $\mathsf{Incr}~A~B$ stands for the incrementalization tuple type $(C: \mathsf{Type},~ A \to B \times C,~ A' \to C \to B' \times C)$):
\[\begin{array}{lll}
	\mathsf{iter} &:& (I : \mathsf{Incr}~A~B) \to A \to \mathsf{List}~A' \to B \times \ctype{I} \\
	\mathsf{iter}~(C, i, d)~x~[] &:=& i~x \\
	\mathsf{iter}~(C, i, d)~x~(x'::xs') &:=& \mathsf{let}~(y, c_1) := \mathsf{iter}~(C,i,d)~x~xs' \\
	&& \mathsf{let}~(y', c_2) := d~x'~c_1 \\
	&& (y\oplus y', c_2)
\end{array}\]

Note that changes are applied back to front.
To represent full reevaluation, we need to compute the input value after the list of changes has been applied, using the $\mathsf{sum}$ function below.
\[\begin{array}{lll}
\mathsf{sum} &:& A \to \mathsf{List}~A' \to A \\
\mathsf{sum}~x~[] &:=& x \\
\mathsf{sum}~x~(x'::xs') &:=& \mathsf{sum}~x~xs' \oplus x'
\end{array}\]
We can now define the \soundness{} property.

\begin{definition}[\Sound{} Incrementalization]\quad\\
	We call an incrementalization $I : \mathsf{Incr}~A~B$ of $f : A \to B$ \emph{\sound{}} if
	\[\forall (x:A)~(xs':\mathsf{List}~A).~(\mathsf{iter}~I~x~xs')_1 = f~(\mathsf{sum}~x~xs').\]
\end{definition}

Finally, we prove showing incrementalization to be consistent (from Definition~\ref{def:incrementalization}) is sufficient for ensure that an incrementalization is \sound{}.

\begin{lemma}\label{lem:incrementalizationIter}\quad\\
	If $~(C, i, d) : \mathsf{Incr}~A~B$ is a consistent incrementalization of $f : A \to B$, then it is a \sound{} with regard to $f$.
\end{lemma}
\begin{proof}
We show the following, stronger, property by induction on $xs'$:
\[\forall x~xs'.~\mathsf{iter}~I~x~xs' = i~(\mathsf{sum}~x~xs').\]
The lemma follows by the first consistency law ($\forall x.~i_1~x = f~x$).
\end{proof}

\subsection{Incrementalization and Combinators Are \Sound{}} %
\label{sec:incrsound}

We now prove that our incrementalization function and the incrementalization combinators produces consistent, hence \sound{}, incrementalizations.

We first prove that the result of the incrementalization function fulfills the consistency laws.

\begin{lemma}\label{lem:incrementalizationTuple}\quad\\
	For an expr. $e$, the tuple $(\ctype{\den{e}}, \init{\den{e}}, \deriv{\den{e}})$ is a consistent incrementalization of $\den{e}$.
\end{lemma}
\begin{proof}
By induction on $e$.
\end{proof}

In our Lean mechanization, this means that \lstinline{Tm.incr} has the following signature (which we already showed above), ensuring that it always returns a correct incrementalization:
\begin{lstlisting}
def Tm.incr : (e:Tm A B) → ConsIncr A.toCS B.toCS e.denote
\end{lstlisting}
\Soundness{} follows immediately from this signature.
The function \lstinline{Tm.incr} is defined by structural recursion on the syntax. 
To ensure that the consistency laws from Definition~\ref{def:incrementalization} are met, we prove them individually for each language construct.
As presented in Figure~\ref{fig:incrRules}, our formalization uses our incrementalization combinators.
Note that whenever we use one of our combinators, we have to prove that all its preconditions are met.
Consider for instance the $\src{id}$ operation (Figure~\ref{fig:incrRules}), which we incrementalize with our \textsc{Self} combinator (Figure~\ref{fig:incrComb}).
Here, we have to prove \textsc{Self}'s precondition $\den{\src{id}}~(x \oplus x') = \den{\src{id}}~x \oplus \den{\src{id}}~x'$, which clearly reduces to the identity.

\begin{theorem}[\thmIncrSoundName]\label{thm:ConsIncrementalization}\quad\\
	For an expression $e$, the tuple $(\ctype{\den{e}}, \init{\den{e}}, \deriv{\den{e}})$ is a \sound{} incrementalization of $\den{e}$.
\end{theorem}
\begin{proof}
	From Lemma~\ref{lem:incrementalizationTuple} and Lemma~\ref{lem:incrementalizationIter}.
\end{proof}

\begin{theorem}[\thmCombSoundName]\label{thm:CorrectCombinators}\quad\\
The combinators \textsc{Triv}, \textsc{Self}, \textsc{Lin}, \textsc{BiLin}, \textsc{Add} are \sound{}.
\end{theorem}
\begin{proof}
We show consistency for each combinator using equational reasoning, from which \soundness{} follows by Lemma~\ref{lem:incrementalizationIter}.
\end{proof}

In Lean, each incrementalization combinator's consistency is expressed similarly to \lstinline{Tm.incr}, namely by having them return values of type \lstinline{ConsIncr}. More precisely, each combinator takes a function \lstinline{f : A.V → B.V} (and some additional data, corresponding to the premises in Figure~\ref{fig:incrComb}) and returns a value of type \lstinline{ConsIncr A B f}.

\subsection{Finite Support}
\label{sec:finsupp}

The above results ensure the extensional correctness of our approach.
In this subsection, we want to discuss an issue relevant for implementations
that use finite data structures like hash maps to represent containers:
If some shapes of the container have infinitely many indices, we need to ensure that
our program only produces containers with \textit{finite support}, i.e., where only a
finite number of indices are mapped to non-zero values.
Accordingly, we give theorems stating under which conditions our built-in
operations preserve the finite-support-property.
One common condition for all finite-support properties is the existence of a default value (called $\varepsilon$ in the following) for each the type a container is applied to. This $\varepsilon$ is used to represent the absence of a value.

To formalize the notion of finite support preservation, we define finitely supported
maps as functions $f : A \to B$ for which there is an $l : \mathsf{List}~A$ such that for all
$k \notin l$, $f~k = \varepsilon$.
We then show that the operators returning containers preserve finite support. We do this by applying their denotational
semantics to finitely supported maps, thereby proving that each operation returns a finitely supported map.

\begin{theorem}[Finiteness Preservation]\label{thm:FinPres}\quad\\
The operators $\src{set}$, $\src{replicate}$, $\src{map}$, $\src{reshape}$,
$\src{filter}$, $\src{zip}$, $\src{tp}$ preserve finite support of containers, under the following preconditions:
(i) For $\src{map}~f$, we require $\den{f}~\varepsilon = \varepsilon$.
(ii) For $\src{replicate}~x$, we require that $x = \varepsilon$.
(iii) For $\src{reshape}~r$, we require that for each index $r~i$ in the input,
there only finitely many indices $j$ in the output such that $r~j = i$.
\end{theorem}

The conditions in the theorem prevent the operations from setting infinitely many non-zero values.
Our relational algebra and tree use cases fulfill these requirements, as we do not define
operations that create infinite relations out of finite ones. The linear algebra and CRDT use cases are not bound by these preconditions as they only use finite data structures.

%% file: implementation.tex
The mechanization of \Lang{} in Section~\ref{sec:correct} is intended to verify our correctness claims, but is not directly usable for practical use cases. In particular, the use of functions to represent data structures makes it difficult to reason about performance. We therefore present an implementation that uses a representation based on hash maps and also provides better ergonomics for writing instances.

Our implementation is based on tagless-final style~\cite{Carette09}. 
This means that rather than defining the abstract syntax as an inductive data type, we define them as a type class, where each function corresponds to a constructor. 
This allows us to define an interpreter by instantiating this type class.

\subsection{Efficient Representation}

We keep our library implementation as close as possible to the formalization.
However, to ensure that all operations are efficiently computable, our implementation differs slightly from our calculus in several details.
Importantly, none of these affect its soundness.

Most notably, while our formalization represents data structures as functions, our library uses hash maps. 
This allows eager computation of the structure's values and enables benefiting from sparsity. 
The latter is especially important when dealing with changes, which are often sparse (e.g., when we only change a few tuples in a relation). We can ensure that sparse changes are applied efficiently by updating data structures in place. In our implementation, this is enabled by Lean's functional-but-in-place semantics~\cite{Reinking20}, which allows destructive updates for linearly used data structures while maintaining purely functional behavior. In other languages, one could
instead use purely functional data structures supporting efficient insertions~\cite{Okasaki99}.

As discussed in Section~\ref{sec:finsupp}, we need to be careful when representing infinite data structures
as (finite) hash tables. Our implementation assumes that the preconditions mentioned in Section~\ref{sec:finsupp} are respected by the programmer (this is not needed when dealing with finite data structures like arrays).

Finally, the formalization requires the indices of the container to have decidable equality, while the implementation requires Boolean equality and a hash function\footnote{The difference between decidable equality and Boolean equality in Lean is that the former needs to provide a proof of equality or inequality, while the latter is just a Boolean-valued function.}.

It is important to note that our correctness argument from Section~\ref{sec:correct} also applies to our implementation.
(i)~The difference between requiring decidable equality and Boolean equality for the indices is also unproblematic for correctness.
This is the case, since the existence of a Boolean-valued equality function follows straightforwardly from the existence of decidable equality.
(ii)~Hash maps are clearly a valid implementation of finitely-supported functions.
Hence, our correctness argument from Section~\ref{sec:correct} also holds for our implementation.

\subsection{Syntactic Sugar}\label{sec:frontend}

While \Lang{} is primarily a formal calculus, we add let-bindings as syntactic sugar.
The core idea is that let-bindings $\mathsf{let}~x=e_1;~e_2$ can be translated into existing \Lang{} combinators,
by duplicating the context, passing one copy to $e_1$, and then passing the result of $e_1$ along with the other copy to $e_2$:
\[\den{\mathsf{let}~x=e_1;~e_2} = \src{dup~;~} (\den{e_1}~\src{\times~id})~\src{;}~\den{e_2} \]

The translation of named variables takes two steps
from named variables to de Bruijn variables, and then from de Bruijn to \Lang{}'s combinators.
The first step translates named variables to de Bruijn indices.
The type context of a named term is initially a list of variable-type pairs.
They become de Bruin contexts by stripping away the variable names.
For example, a named context $x:A,~y:B$ is translated to a de Bruijn context $A, B$.
Correspondingly, variable occurrences in terms are translated to indices by looking up their position in the context.
Hence, $x:A, y:B \vdash~ y : B$ is translated to $A,B \vdash~ \overline{1} : B$,
as $y$ was the second variable in the context and we start counting at 0.

The second step translates de Bruijn terms to combinatorial terms.
This step completely eliminates contexts and replaces them by nested products.
That is, $A, B$ is translated to $A \times B$ and $A,B \vdash e : C$ to $e' : A \times B \to C$.
Correspondingly, it translates de Bruijn variables to projections.
For instance, $A,B \vdash~ \overline{1} : B$ is translated to $\src{snd} : A \times B \leadsto B$.

\paragraph{Example} We give an example how the matrix-vector multiplication and the dense layer
from Section~\ref{sec:ex} looks like using our syntactic sugar.
Here, \lstinline{NTm} represents a named term, \lstinline{#} is function application,
and \lstinline{[m, v]} is used to construct the argument list for named terms like \lstinline{mvmul}.

\begin{lstlisting}
def mvmul : NTm [("m", V n (V m Float)), ("v", V m Float)] (V n Float) :=
  ⟪ map sum # (map2 (map2 mul) # (replicate # v, m)) ⟫
def dense : NTm [("m", V n (V m Float)), ("b", V n Float), ("x", V m Float)] (V n Float) :=
  ⟪ map relu # map2 add # (mvmul # [m, x], b) ⟫
\end{lstlisting}

%% file: application.tex
In this section, we describe the instances we have implemented, applying \Lang{} to linear algebra,
relational algebra, trees, and CRDTs.
We thereby show how \Lang{} enables incrementalization in new areas without reinventing the wheel.
We perform proof-of-principle performance evaluation for linear algebra, relational algebra, and trees.
We discuss the results of the linear algebra example (a dense neural network layer) here. For
the other applications, we only give a brief description and provide more details in Appendix~\ref{sec:evaldetails}.

\subsection{Linear Algebra}

\begin{figure}
\begin{minipage}{.55\linewidth}
\captionsetup*{type=table}
\caption{Linear Algebra operations in \Lang{}.}
\label{tbl:laops}
\small
\centering
\begin{tabular}{@{}rcl@{}} \toprule
Operation & Lin. Alg. & \Lang{} \\ \midrule
Vector sum & $v_1 + v_2$ & $\src{map2~+}$ \\
Matrix sum & $M_1 + M_2$ & $\src{map2~(map2~+)}$ \\
Hadamard product & $v_1 \circ v_2$ & $\src{map2~}\ext{*}$ \\
Dot product & $v_1 \cdot v_2$ & $\src{map2~}\ext{*}~\src{;}~\ext{sum}$ \\
Scalar-vector product & $c \cdot v$ & $\der{svmul}$ \\
Matrix-vector product & $M~v$ & $\der{mvmul}$ \\
Matrix-matrix product & $M_1~M_2$ & $\der{mmmul}$ \\
\bottomrule
\end{tabular}

\quad \\ %
where $\der{svmul} := \src{(replicate}~n\src{ \times id) ; map2~}\ext{*}$, \\
and $\der{mmmul} := \src{replicate}~n \src{\times tp; map2~}\der{mvmul}\src{;tp}$

\end{minipage}\hfill\begin{minipage}{.45\linewidth}

\quad\\

\includegraphics[width=\linewidth]{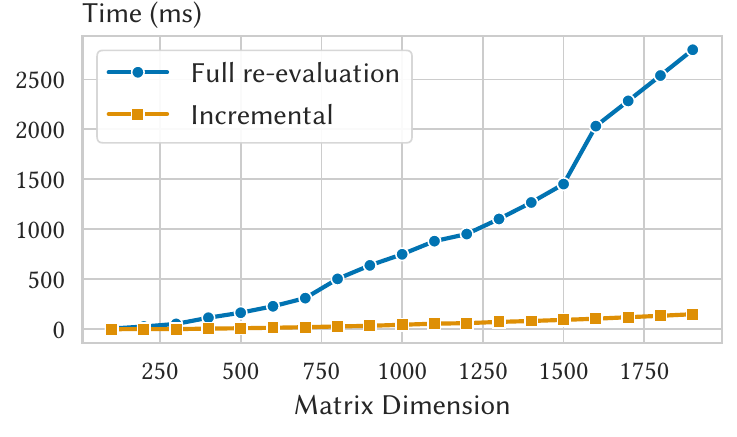}
\caption{Linear Algebra benchmark: Full evaluation vs. incremental update on $\der{dense}$.}
\label{fig:eval}
\end{minipage}
\end{figure}

In Section~\ref{sec:laExtension}, we have presented the implementation of the syntax and denotational semantics of the linear algebra instance.
In this section, we present the incrementalization of the instance's operations and discuss how common linear algebra operations are encoded in \Lang{}. The implementation of the incrementalizations is straightforward, thanks to the use of incrementalization combinators
(\lstinline{Map.sum} sums the values of a container, represented as a map):
\begin{lstlisting}
def mul := triv (λ (x, y) => x * y)
def relu := triv (λ x => if x > 0 then x else 0)
def sum := self Map.sum Map.sum
\end{lstlisting}

In Table~\ref{tbl:laops}, we show how linear algebra operations can 
be expressed in \Lang{} (again we use $\src{map2}~f$ as an abbreviation for $\src{zip; map}~f$).
We assume that $\der{mvmul}$ is defined like in Section~\ref{sec:ex}.

\paragraph{Evaluation}

For a preliminary qualitative assessment of the speed-up provided by our approach, we perform a small benchmark using the dense neural network layer $\der{dense}$ running example defined earlier. We compare the time needed to fully evaluate the dense layer with the time needed to update the output when 1\% of the input elements change. 
The benchmarked input vector sizes range from $n=100$ to $n=2000$ (the size of the weight matrix is always $n \times n$). Every
time measurement is best-of-5\footnote{We ran all experiments single-threaded on a ThinkPad T14 Gen 4 with AMD Ryzen 7 PRO 7840U CPU (\SI{3.3}{GHz}) with Linux kernel 6.12.63 and Lean 4.22.0-rc4.}.
The results shown in Figure~\ref{fig:eval} indicate a clear speed-up of the incrementalization compared to reevaluation. These results are meant as a "proof-of-principle" regarding the viability of our approach; please note that our current implementation uses an unoptimized interpreter. Hence, at the moment, we do not claim to be competitive with state-of-the-art frameworks that use compilation and specialized linear algebra libraries.
Note that the cache needed for incrementalization can be read off from our definition of the incrementalization translation.
Specifically, the cache type for $\der{dense}$ is (ignoring caches of Unit type) $(\mathbb{R}^2)^{n\cdot m} \times \mathbb{R}^n$ for an $m\times n$ matrix, so the cache contains $2 n m + n$ numbers.

We also evaluate how performance declines as changes get bigger (see more in Section~\ref{sec:laeval}).
For our example, incremental evaluation remains faster than reevaluation only as long as at most 70\% of the input changes.
Note that our evaluation uses the same hash-map-based representation for matrices irrespective of the density;
the use of a dense representation would likely mean the cut-off point at which reevaluation becomes faster is lower than 70\%.

\subsection{Relational Algebra}

We follow DBSP~\cite{Budiu23} in representing a database table as a map from tuples to integers, representing the multiplicity of the tuple in the table (negative integers are used to represent the removal of tuples). 
That is, we have as base type $\beta := \mathbb{Z}$ and as the shape type $S$ a collection of schemata. 
Concretely, our implementation defines a schema to be either an integer, a string,  or a pair of two schemata, but this could be changed without affecting the rest of the instance. 
Note the contrast to the linear algebra example: 
While there, the values of the vector were the values of the container, here, the values of the relation are the \textit{indices} of the container. 
More formally, a table with schema $s$ is represented as a map from $\pos{s}$ to $\mathbb{Z}$, where the indices $\pos{s}$ are the values of the relation.
A change then amounts to adding/removing an arbitrary number of tuples to/from the table.

We formally describe the instance in Figure~\ref{fig:relalginst}. The Cartesian product $\ext{\otimes}$ takes two relations with schemata $s$ and $s'$ and returns a relation with the product schema $s,s'$. It is the only instance operation that is not defined using an incrementalization combinator. Its incrementalization can be seen as the discrete derivative of multiplication:
\[ ((a + a') \cdot (b + b')) - (a \cdot b) = a \cdot b' + a' \cdot b + a' \cdot b'. \]

In Table~\ref{tbl:raops}, we show how several relational algebra operations~\cite{Codd70} can be expressed in \Lang{}.
Defining relational algebra as an instantiation of \Lang{} is simpler than defining its incrementalization from scratch for three reasons. First, of the RA operations, only the Cartesian product has to be explicitly defined as an additional operation, while the other can be defined as derived operations in \Lang{} and hence require no additional proof effort. Second, no incrementalizations need to be written by hand, thanks to the use of incrementalization combinators. Third, incrementalization combinators reduce the proof effort, as it is only required to show the bilinearity of $\der{\ext{\otimes}}$ and some trivial properties of integer addition and subtraction.

\begin{figure}
\begin{minipage}{.5\textwidth}
\centering
$\begin{array}{lll}
\ext{*} &:& \src{b \times b} \leadsto \src{b} \\
x~\den{\ext{*}}~y &:=& x \cdot y \\
\den{\ext{*}}' &:=& \textsc{Triv}~\den{\ext{*}} \\[.5em]
\ext{-} &:& \src{b \times b} \to \src{b} \\
x~\den{\ext{-}}~y &:=& x - y \\
\den{\ext{-}}' &:=& \textsc{Lin}~\den{\ext{-}} \\[.5em]
\ext{\otimes} &:& \src{F}_s~\src{b}~\src{\times}~\src{F}_{s'}~\src{b} \leadsto \src{F}_{s,s'}~\src{b} \\
R~\den{\ext{\otimes}}~S &:=& \lambda (i, j).~x~i \cdot y~i \\
\den{\ext{\otimes}}' &:=& \textsc{BiLin}~\den{\ext{\otimes}}
\end{array}$
\caption{Relational algebra instance operations.}
\label{fig:relalginst}
\end{minipage}%
\begin{minipage}{.5\textwidth}
\captionsetup*{type=table}
\caption{Relational Algebra operations in \Lang{}.}
\label{tbl:raops}
\centering
\begin{tabular}{@{}rcl@{}} \toprule
                  Operation & Rel. Algebra & \Lang{} \\ \midrule
Cartesian prod. & $R \times S$ & $\ext{\otimes}$ \\
Join & $R \bowtie_p S$ & $\ext{\otimes}\src{; filter~}p$ \\
Selection & $\sigma_p(R)$ & $\src{filter~}p$ \\
Union & $R \cup S$ & $\src{map2~+}$ \\
Difference & $R - S$ & $\src{map2~}\ext{-}$ \\
Intersection & $R \cap S$ & $\src{map2~}\ext{*}$ \\
\bottomrule
\end{tabular}
\end{minipage}
\end{figure}

\paragraph{Dictionaries}

To define some aggregations like group-by, we use both relations and \textit{dictionaries}.
Dictionaries~\cite{Shaikhha22} are a special case of containers, where the shape type has a single element, and hence the index type is not dependent on the shape.
Specifically, dictionaries for a given key type $K$ can be seen as a functor
mapping $V$ to $K \to V$, where the functions are assumed to be finite.
Hence, for each $K$, we define a container $\src{D}^K$ with a single shape $\star$
and index type $\pos{\star} := K$. Our container semantics means that $\src{D}^K$ applied to a type $A$ denotes $\den{\src{D}^K_\star~A} = (\pos{\star} \to A) = (K \to A)$.

\paragraph{Aggregations}
Our first example of an aggregation is $\ext{count} : \src{F}_s~\src{b} \leadsto \src{b}$, which counts the number of rows
in a relation. It is incrementalized using the \textsc{Self} combinator.
Similarly, we can define summation, maximum, and minimum, which we omit here.
We further have an operation
$\ext{groupBy} : (\pos{s} \to K) \to \src{F}_s~\src{b} \leadsto \src{D}^K_\star~(\src{F}_{s}~\src{b})$,
inspired by DBSP.
In addition to the relation container it uses the dictionary container defined above.\footnote{
using the given type for $\ext{groupBy}$, instead of $\ext{groupBy} : (\pos{s} \to \pos{s'}) \to \src{F}_s~\src{b} \leadsto \src{F}_{s'}~(\src{F}_{s}~\src{b})$ allows grouping by any type $K$, not just types expressible with the schema $s'$.}
With $\ext{count}$ and $\ext{groupBy}$ defined, we can construct the derived operation
$\der{proj} := \ext{groupBy}~f \src{~;~map}~\ext{count}$. We describe a proof-of-principle
evaluation of incremental $\der{proj}$, showing a speed-up over reevaluation, in Appendix~\ref{sec:raeval}.

\subsection{Trees}

The natural approach for representing tree-like structures as containers is as (partial) maps
from paths to scalar values, similar to XPath and JSONPath.
We want changes to allow adding and removing tree nodes.
Our approach represents dynamically-shaped trees
by letting the type of ``shapes'' be the Unit type (similar to our container representation, and to how dynamically-typed languages
can be represented as uni-typed languages), and letting the index set be the set of all possible paths.
The fact that the mapping from paths to values is partial is reflected in the fact that the
type of values has a ``null'' value, which ensures that our operations are still formally
total (but they may produce structures containing null values)

\paragraph{XQuery Example}

We implement an example inspired by Query 1 from the XQuery Use Cases document\footnote{\url{https://www.w3.org/TR/xquery-use-cases}}, which has
also been studied by~\citet{Matsuda23}.
In this example, we query a bibliography for books published by Addison-Wesley after 1991, and return their titles and years.
As alluded to above, our incremental program expects the books to be represented as maps from
paths to scalar values. The conversion of a normal tree into a map and back is implemented in Lean.
The data and query are shown in Appendix~\ref{sec:jsoneval}.

\paragraph{Folds on Rose Trees}

We now describe how we can incrementalize folds on rose trees, which are trees where each node can have an arbitrary number of children.
Similar to the XQuery example, we define rose trees first as an inductive type in Lean, and then implement conversion to maps from paths to scalar values.
A fold then takes a function $f$ which is applied to each path-value pair,
a binary function $m$ which merges the results of $f$ for each path-value pair,
and a neutral element $z$ for $m$.
For the definition to be well-defined and incrementalization to be practical, we further assume that $m$
is associative, commutative and is linear (i.e. its own incrementalization).
We can then define use folds to define aggregations over trees like height, size, maximum, and sum.
Given our assumptions, fold is self-incrementalizable. We evaluate summation over trees with
branching factor 2 and up to depth 20 in Appendix~\ref{sec:roseeval}, showing a speed-up over reevaluation.

\subsection{CRDTs}

State-based conflict-free replicated data types (CRDTs)~\cite{Shapiro11, Almeida25} are
a special kind of replicated data types, whose state
fulfills the semilattice laws.
We demonstrate incrementalization on a classic CRDT example, a global counter (GCounter) shared among network participants $\mathbb{I}$.
The state is a vector of natural numbers, one per node. $\der{inc}$ increments the $i$-th entry, $\der{value}$ sums all entries, and $\der{merge}$ takes element-wise maximum.

\begin{figure}
$S      := \mathsf{Unit} \qquad\qquad
\pos{\star} := \mathbb{I} \qquad\qquad
\beta   := \mathbb{N} \qquad\qquad
+       := +_\mathbb{N} \qquad\qquad
0       := 0_\mathbb{N}$\\[1em]
\noindent
$\begin{array}{lll}
\den{\ext{max}}~x~y &:=& \mathsf{if~}x<y\mathsf{~then~}y\mathsf{~else~}x \\
\den{\ext{max}}'    &:=& \textsc{Triv}~\den{\ext{max}} \\
\den{\ext{sum}}~x   &:=& \Sigma_i (x~i) \\
\den{\ext{sum}}'    &:=& \textsc{Lin}~\den{\ext{sum}}
\end{array}$\qquad\qquad
$\begin{array}{lll}
\der{value}  &:=& \ext{sum} \\
\der{incNat} &:=& \src{\langle id,~cst~1\rangle ; +} \\
\init{\der{inc}}  &:=& \src{\langle \init{get} ; \der{incNat},~id \rangle ; \init{set}} \\
\der{merge}  &:=& \src{map2~}\ext{max} \\
\end{array}$
\caption{GCounter CRDT instance, including derived operations.}
\label{fig:crdt}
\end{figure}

Figure~\ref{fig:crdt} shows the GCounter implementation in \Lang{} with one shape ($\star$), whose index set is the set of participants.
Incremental $\der{value}$ updates the sum without re-iterating the vector, and incremental $\der{merge}$ operates only on sparse changes, reducing computation and bandwidth.
This yields a \textit{delta-state CRDT}~\cite{Almeida18}: normal evaluation behaves as a state-based CRDT (transmitting full state), while incremental evaluation behaves as a delta-state CRDT (transmitting only changes). Traditionally, delta-state CRDTs require separate manual implementations; we derive both from a single \Lang{} specification.

%% file: related.tex
We discuss the relationship of our approach to domain-specific incrementalization frameworks,
to the incremental $\lambda$-calculus, and to other works on containers and data-type-generic incrementalization.
We close by discussing the relationship of incrementalization to streaming.

\paragraph{Domain-specific incrementalizations}

DBSP~\cite{Budiu23} is a language for incremental stream computations.
They expect types to form commutative groups, whereas we use change structures, which are more general.
For instance, we use $\mathbb{N}$, which is not a group, in our CRDT example. 

DBSP models updates using causal stream operators (CSOs). A CSO is a function from streams to streams, where the $n$-th element of the output only depends on the first $n$ elements of the input. While the CSO abstraction is theoretically appealing, it abstracts away implementation details crucial for efficient incrementalization, such as the structure and size of the cache required to process updates. Further, CSOs are modeled to operate on whole streams (which are represented as functions), making the memory requirements and computational cost unclear.

In our work, we represent as a function that explicitly manages a state, in the style of a Mealy machine. More precisely:
\[\deriv{\den{f}} : A' \to \ctype{\den{f}} \to B' \times \ctype{\den{f}}\]
This approach assigns each program a specific cache type $\ctype{\den{f}}$, making explicit not only the existence of caching but also the concrete structure of the cache for each operation. In this work, our incrementalization systematically transforms every operation into a concrete Mealy machine, ensuring that our formalization closely mirrors an efficient implementation.

Based on their theory, Budiu et al. implement a database query language for incremental view maintenance. 
Like our relational algebra instance, they model relations as maps to $\mathbb{Z}$. 
We also apply our approach to other use cases and provide a set of generic data transformations and combinators.

Differential Dataflow (DD)~\cite{McSherry13} is an earlier approach that, like DBSP, formalizes changes using commutative groups. 
In contrast to DBSP (and our approach), DD focusses on verifying a library of
incremental operators rather than on incrementalizing existing programs. DD also does not formalize how new operations can be added while maintaining soundness~\cite{Budiu23}.

LAGO~\cite{Shaikhha20} is a framework for incrementalizing linear algebra programs. 
It uses static analysis and matrix factorization to increase the efficiency of change representation and computation. 
In contrast, we focus on providing a mechanized formalization of incrementalization which supports a wider class of programs.
It could be interesting to try to integrate their optimizations~into~our~framework.

\paragraph{Incremental $\lambda$-calculus} A series of papers has proposed the incremental $\lambda$-calculus~\cite{Cai14} (I$\lambda$C) as a language for incrementalization involving higher-order functions.
The formalizations in those papers focus on incrementalizing $\lambda$-abstraction and function application. Specific data structures need to be defined separately and are not part of the formalization.
The original I$\lambda$C~\cite{Cai14} is a simply-typed $\lambda$-calculus where each type is equipped with a \textit{change structure}, associating a set of changes to every value.
The I$\lambda$C was extended to support caching by \citet{Giarrusso19}, using an untyped language.
Recently, a typed version of the cached I$\lambda$C without first-class functions was defined using tagless-final style by \citet{Matsuda23}, with a definition of incremental functions similar to ours, but without proofs.
Our approach and the I$\lambda$C are both incremental programming languages, but the feature sets are mostly disjoint: 
\Lang{} provides data-type-generic operations such as $\src{map}$ and $\src{reshape}$ and incrementalization combinators.
In contrast, the I$\lambda$C focuses on first-class functions, but does not provide any built-in operations in the core language. Adding new data types to the I$\lambda$C requires implementing all operations from scratch.

\paragraph{Caching and higher-order functions}
Our approach uses cache-transfer-style and supports well-typed incrementalization but no higher-order functions. Higher-order approaches either
do not support caches~\cite{Cai14}, or are untyped \citet{Giarrusso19}.
One obstruction is that incremental programs taking arguments of function type can’t be safely typed because the closure’s environment is not statically known, and therefore neither is the cache type (as noted by \citet{Matsuda23}, who, like us, implement a first-order language).
For automatic differentiation, it has been proposed to apply defunctionalization before AD, rather than to perform AD on first-class functions \cite{Smeding24}. A similar approach seems promising for incrementalization.

\paragraph{Containers}
\citet{Wagner14} also studies a notion of changes on containers, but in the context of symmetric lenses.
He uses containers to develop a library for composing lenses over data types. The difference to our work is that
lenses are about maintaining consistency between to data structures in the face of changes, rather than in efficiency gains. Additionally, Wagner's approach is aimed at constructing lenses manually, while we are interested in automatically incrementalizing programs.
\citet{Bohler24} performed a preliminary investigation on change structures and containers for incrementalization,
but provided incrementalizations a very small number of operations, all of which are cache-free.
In contrast, the current paper presents a program transformation for cached incrementalization and a modular approach
with a generic core language and domain-specific extensions.
Their work also does not feature incrementalization combinators or case studies.

\paragraph{Other approaches to incrementalization}

The \textsc{Adapton} library~\cite{Hammer14} implements incrementalization by building up a dependency graph between computations, with mutable references at the leaves. This approach is independent of data structures, and is able to compute changes lazily on demand. However, it incurs the overhead of building up a graph at runtime. More importantly, the incrementalization is all-or-nothing: If a value changes, every %
computation that depends on it needs to be fully recomputed. 
We instead aim to handle small changes on data structures efficiently.

\citet{Morihata18} describes an approach to incremental computing which is generic over inductive data structures, deriving elegant incrementalizations for functional programs.
They define data structures as algebraic data types with a list of constructors, %
while we define data structures as containers.
Values of data structures defined as algebraic data types are trees of constructor applications.
Hence, in order to access any values, a functional program has to recursively traverse the tree.
Our container-based definition separates shape and content.
This allows us to back the content by an array with constant-time access.
This is important 
to support random access well.
Further, Morihata does not provide an automatic program transformation.
It is therefore not clear how to use the approach other than to derive incrementalizations by hand.

\paragraph{Streaming}

As noted in Section~\ref{sec:fw}, there is a close correspondence between incrementalization and streaming.
Incrementalization takes a ``batch'' operation which takes entire values and turns it into an operation which iteratively
receives input changes and produces an output change each time. This notion corresponds to causal stream functions (as used by~\citet{Budiu23}) and Mealy machines (as used by us).

The Flo language~\cite{Laddad25} formalizes correctness properties for streaming programs:
Streaming progress means that a program does not block because it waits for a non-terminating stream to terminate,
and eager execution means that a program produces deterministic outputs, no matter if executed incrementally or in large batches.
In our case, progress is ensured by the definition of incremental functions as Mealy machines,
and eager execution roughly corresponds to our incrementalization correctness theorem.

The view of incrementalization-as-streaming can also be applied to work on semilattices for
distributed programming~\cite{Kuper13,Meiklejohn15,Conway12,Milano19}:
A semilattice over $A$ can also be interpreted as a homomorphism from a set of
inflationary updates to $A$. If we assume that a change on this set consists of adding a new update, then
incrementalizing the homomorphism yields a streaming program over a stream of updates.

%% file: conclusion.tex
In this paper we present a modular and generic approach to incrementalization which covers a wide range of data structures.
We present it in form of a core calculus called \Lang{}.
Our calculus is not tied to specific data structures or operations, 
but supports instantiation to any data structure we can decompose into a shape and a set of indices.
Meanwhile, our framework allows more fine-grained incrementalizations than selective reevaluation, which purely focus on the dependencies between operations, omitting the incrementalization of the operations themselves.

\Lang{} is a functional programming language
that supports a set of common generic functional operations, such as $\src{map}$ and $\src{filter}$.
Our goal is to automate incrementalization of complex computations.
Hence, our technique focuses on the composition of derivatives of the operations occurring in a program.
Our framework derives a program's incrementalization by mapping operations to their derivatives and by composing those.
The result is an incrementalized program which receives an \emph{input change} and uses it to update the output.
For many operations, however, the output change does not only depend on the input change, 
but also on the former output itself, e.g., for multiplication.
We address this issue by reusing and updating former results via caching~\cite{Giarrusso19}.

Note that our calculus ships with incrementalizations for the basic operations.
When instantiating it to a specific domain, we require the user to provide incrementalizations for the additional domain-specific operations.
We propose a set of incrementalization combinators to simplify this task.

Our work is mechanized and verified in the Lean theorem prover.
Through the use of intrinsic typing and our incrementalization combinators, 
we guarantee that all incrementalizations are correct by construction.
We codified these claims in our
\thmIncrSoundName~\ref{thm:ConsIncrementalization} and \thmCombSoundName~Theorems~\ref{thm:CorrectCombinators}.

To demonstrate the versatility of our framework, we implemented several case studies, and conducted corresponding performance 
evaluation observing significant speed-ups over reevaluation.

\paragraph{Future work}

Our approach uses a first-order language. 
As mentioned above, it would be interesting to support first-class functions in \Lang{} via defunctionalization.

\Lang{} might also be well suited to increment \textit{imperative} programs.
In particular, the approach of \citet{Kumar25} models imperative programs as functions from an input state to an output state.
This seems to fit well with \Lang{}, where every program has a type of the form $A \leadsto B$.
We could represent states consisting of multiple variables as nested tuples, and  express sequential execution of statements through sequential composition ($\src{;}$).
A further syntactic extension could come from the Arrow calculus~\cite{Lindley10}, which allows a limited form of function abstraction and application,
translating them into a combinatorial language that has striking similarities with ours.

Our current \Lang{} implementation realizes incremental execution through an interpreter. 
We could improve the performance by implementing our incrementalization as a source-to-source transformation and by adding a compiler backend.

Finally, in this work, we consider changes to impact only individual elements of a data structure.
Accordingly, we assume shapes to be static.
Insights from Morihata's~\cite{Morihata18} work on incremental algebraic data types could allow for incrementalization over data with dynamic shapes.

%% file: data.tex
We supply an artifact~\cite{Boehler26} in the form of a Virtual Box 7.2 image, containing our Lean proof mechanization and implementation. The artifact describes which definitions and theorems in the paper correspond to which parts of the code. It allows running the Lean type checker, verifying that the proof checker accepts our proof, and to run our benchmarks. We also provide instructions for writing new instances
of containers and incremental operations.

%% file: sums.tex
In Figure~\ref{fig:sumcs}, we define the change structures for containers, products, and sums.
The change structures for containers and products apply their operations elementwise; the case for sums is more involved,
featuring a custom type of sum changes.
The local change $\mathsf{cl}~x'$ updates a value $\iota_1~x$ to $\iota_1~(x \oplus x')$
and leaves it unchanged otherwise ($\mathsf{cr}$ does the same for $\iota_2$). The change $\mathsf{sl}~x$ replaces the old value
with a new value $x$. This non-local change is used for ensuring completeness, i.e. the ability to change any value into any other value: Local changes ($\mathsf{cl}$, $\mathsf{cr}$) do not allow us to go from $\iota_1~x$ to $\iota_2~y$, as they leave the constructor
unchanged.
We also add a $\mathsf{null}$ change, which is used for incrementalizing $\src{||}$.

The incrementalizations for the sum operations are given in Figure~\ref{fig:sumincr}.

\begin{figure}
\[\begin{array}{lll}
\den{\src{F}_s~A} &:=& \pos{s} \to \den{A} \\
\den{\src{F}_s~A}' &:=& \pos{s} \to \den{A}' \\
f \oplus g &:=& \lambda i.~f~i \oplus g~i \\
f \ominus g &:=& \lambda i.~f~i \ominus g~i \\
\\
\den{A \times B} &:=& \den{A} \times \den{B} \\
\den{A \times B}' &:=& \den{A}' \times \den{B}' \\
(x, y) \oplus (x', y') &:=& (x \oplus x', y \oplus y') \\
(x, y) \ominus (x', y') &:=& (x \ominus x', y \ominus y') \\

\\
\den{A + B} &:=& \den{A} + \den{B} \\
\den{A + B}' &:=& \mathsf{cl}~\den{A}' \mid \mathsf{cr}~\den{B}' \mid \mathsf{sl}~\den{A}
\mid \mathsf{sl}~\den{B} \mid \mathsf{null} \\
\iota_1~x \oplus \mathsf{cl}~x' &:=& \iota_1~(x \oplus x') \\
\iota_1~x \oplus \mathsf{cr}~y' &:=& \iota_1~x \\
\iota_2~y \oplus \mathsf{cl}~x' &:=& \iota_2~y \\
\iota_2~y \oplus \mathsf{cr}~y' &:=& \iota_2~(y \oplus y') \\
x \oplus \mathsf{sl}~y &:=& \iota_1~x \\
x \oplus \mathsf{sr}~y &:=& \iota_2~x \\
x \oplus \mathsf{null} &:=& x \\
\iota_1~x \ominus \iota_1~y &:=& \mathsf{cl}~(x \ominus y) \\
\iota_1~x \ominus \iota_2~y &:=& \mathsf{sl}~x \\
\iota_2~x \ominus \iota_1~y &:=& \mathsf{sr}~x \\
\iota_2~x \ominus \iota_2~y &:=& \mathsf{cr}~(x \ominus y) \\
\end{array}\]
\caption{Change structures.}
\label{fig:sumcs}
\end{figure}

\begin{figure}
\[\begin{array}{llll}
\ctype{\den{\src{fuse}_A}} &:=& \den{A} + \den{A} & \\
\init{\den{\src{fuse}}}~(\iota_1~x) &:=& (x, \iota_1~x) & \\
\init{\den{\src{fuse}}}~(\iota_2~x) &:=& (x, \iota_2~x) & \\
\deriv{\den{\src{fuse}}}~(\mathsf{cl}~x')~(\iota_1~x) &:=& (x', \iota~(x \oplus x')) & \\
\deriv{\den{\src{fuse}}}~(\mathsf{cr}~y')~(\iota_1~x) &:=& (x \ominus x, \iota_1~x) & \\
\deriv{\den{\src{fuse}}}~(\mathsf{cl}~x')~(\iota_2~y) &:=& (y \ominus y, \iota_2~y) & \\
\deriv{\den{\src{fuse}}}~(\mathsf{cr}~y')~(\iota_2~y) &:=& (y', \iota_2~(y \oplus y')) & \\
\deriv{\den{\src{fuse}}}~(\mathsf{sl}~x)~(\iota_1~x_0) &:=& (x \ominus x_0, \iota_1~x) & \\
\deriv{\den{\src{fuse}}}~(\mathsf{sl}~x)~(\iota_2~y_0) &:=& (x \ominus y_0, \iota_1~x) & \\
\deriv{\den{\src{fuse}}}~(\mathsf{sr}~y)~(\iota_1~x_0) &:=& (y \ominus x_0, \iota_2~y) & \\
\deriv{\den{\src{fuse}}}~(\mathsf{sr}~y)~(\iota_2~y_0) &:=& (y \ominus y_0, \iota_2~y) & \\
\deriv{\den{\src{fuse}}}~\mathsf{null}~(\iota_1~x) &:=& (x \ominus x, \iota_1~x) & \\
\deriv{\den{\src{fuse}}}~\mathsf{null}~(\iota_2~y) &:=& (y \ominus y, \iota_2~y) & \\
& \\
\ctype{\den{\src{distr}_{A,B,C}}} &:=& \den{A} \times (\den{B} + \den{C}) & \\
\init{\den{\src{distr}}}~(x, \iota_1~y) &:=& (\iota_1~(x,y), (x, \iota_1~y)) & \\
\init{\den{\src{distr}}}~(x, \iota_2~y) &:=& (\iota_2~(x,y), (x, \iota_2~y)) & \\
\deriv{\den{\src{distr}}}~(x', \mathsf{cl}~y')~(x, \iota_1~y) &:=& (\mathsf{cl}~(x', y'), (x \oplus x', \iota_1~(y \oplus y'))) & \\
\deriv{\den{\src{distr}}}~(x', \mathsf{cr}~y')~(x, \iota_1~y) &:=& (\mathsf{cl}~(x', y \ominus y), (x \oplus x', \iota_1~y)) & \\
\deriv{\den{\src{distr}}}~(x', \mathsf{cl}~y')~(x, \iota_2~y) &:=& (\mathsf{cr}~(x', y \ominus y), (x \oplus x', \iota_2~y)) & \\
\deriv{\den{\src{distr}}}~(x', \mathsf{cr}~y')~(x, \iota_2~y) &:=& (\mathsf{cr}~(x', y'), (x \oplus x', \iota_2~(y \oplus y'))) & \\
\deriv{\den{\src{distr}}}~(x', \mathsf{sl}~y)~(x, \_) &:=& (\mathsf{sl}~(x \oplus x', y), (x \oplus x', \iota_1~y)) & \\
\deriv{\den{\src{distr}}}~(x', \mathsf{sr}~y)~(x, \_) &:=& (\mathsf{sr}~(x \oplus x', y), (x \oplus x', \iota_2~y)) & \\
\deriv{\den{\src{distr}}}~(x', \mathsf{null})~(x, \iota_1~y) &:=& (\mathsf{cl}~(x', y \ominus y), (x \oplus x', \iota_1~y)) & \\
\deriv{\den{\src{distr}}}~(x', \mathsf{null})~(x, \iota_2~y) &:=& (\mathsf{cr}~(x', y \ominus y), (x \oplus x', \iota_2~y)) & \\
& \\
\ctype{\den{f~\src{||}_{A_1, A_2, B_1, B_2}~g}} &:=& \ctype{\den{f}} \times \den{B_1} + \ctype{\den{g}} \times \den{B_2} & \\
\init{\den{f~\src{||}~g}}~(\iota_1~x) &:=& (\iota_1~y, \iota_1~(c, y)) & \text{where } (y, c) := \init{f_i}~x \\
\init{\den{f~\src{||}~g}}~(\iota_2~x) &:=& (\iota_2~y, \iota_2~(c, y)) & \text{where } (y, c) := \init{g_i}~x \\
\deriv{\den{f~\src{||}~g}}~(\mathsf{cl}~x')~(\iota_1~(c, y)) &:=& (\mathsf{cl}~y', \iota_1~(c', y \oplus y')) & \text{where } (y', c') := \deriv{f_i}~x'~c \\
\deriv{\den{f~\src{||}~g}}~(\mathsf{cr}~x')~(\iota_1~(c, y)) &:=& (\mathsf{null}, \iota_1~(c, y)) & \\
\deriv{\den{f~\src{||}~g}}~(\mathsf{cl}~x')~(\iota_2~(c, y)) &:=& (\mathsf{null}, \iota_2~(c, y)) & \\
\deriv{\den{f~\src{||}~g}}~(\mathsf{cr}~x')~(\iota_2~(c, y)) &:=& (\mathsf{cr}~y', \iota_2~(c', y \oplus y')) & \text{where } (y', c') := \deriv{g_i}~x'~c \\
\deriv{\den{f~\src{||}~g}}~(\mathsf{sl}~x)~(\iota_1~(c, y_0)) &:=& (\mathsf{cl}~(y \ominus y_0), \iota_1~(c, y)) & \text{where } (y, c) := \init{f_i}~x \\
\deriv{\den{f~\src{||}~g}}~(\mathsf{sl}~x)~(\iota_2~\_) &:=& (\mathsf{sl}~y, \iota_1~(c, y)) & \text{where } (y, c) := \init{f_i}~x \\
\deriv{\den{f~\src{||}~g}}~(\mathsf{sr}~x)~(\iota_1~\_) &:=& (\mathsf{sr}~y, \iota_2~(c, y)) & \text{where } (y, c) := \init{g_i}~x \\
\deriv{\den{f~\src{||}~g}}~(\mathsf{sr}~x)~(\iota_2~(c, y_0)) &:=& (\mathsf{cr}~(y \ominus y_0), \iota_2~(c, y)) & \text{where } (y, c) := \init{g_i}~x \\
\deriv{\den{f~\src{||}~g}}~\mathsf{null}~x &:=& (\mathsf{null}, x) & \\
\end{array}\]
\caption{Incrementalization for sum operators.}
\label{fig:sumincr}
\end{figure}

%% file: evaldetails.tex
We show the library implementation and evaluation results for our case studies.
All our time measurements best-of-5.
The benchmarking machine has an AMD Ryzen 7 PRO 7840U CPU with 16 cores and 32 threads
and runs Debian 13 trixie with kernel 6.12.63. The Lean version is 4.22.0-rc4.

\subsection{Linear Algebra}
\label{sec:laeval}

For the linear algebra evaluation, we benchmark both matrix-vector multiplication and a dense layer (matrix-vector multiplication with bias and ReLU activation).
We apply changes affecting 1\% of the input vector's elements.
The results are plotted in Fig.~\ref{fig:laplotdense} and Fig.~\ref{fig:laplot}.

\begin{figure}
\begin{subfigure}{0.48\linewidth}
\includegraphics[width=\linewidth]{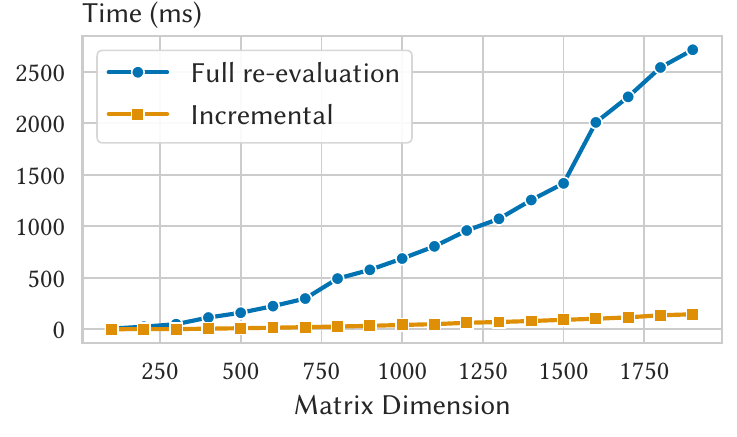}
\caption{Execution time vs. matrix dimension (linear scale).}
\label{fig:laplot-linear}
\end{subfigure}
\hfill
\begin{subfigure}{0.48\linewidth}
\includegraphics[width=\linewidth]{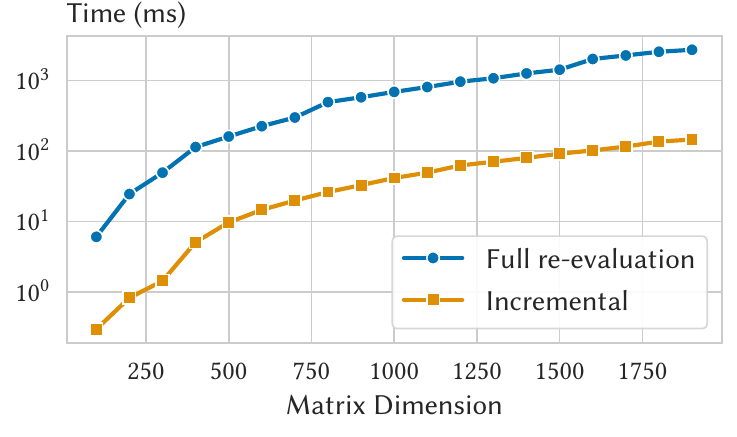}
\caption{Execution time vs. matrix dimension (logarithmic scale).}
\label{fig:laplot-log}
\end{subfigure}
\caption{Matrix-vector multiplication benchmark execution time vs. matrix dimension.}
\label{fig:laplot}
\end{figure}

\begin{figure}
\begin{subfigure}{0.48\linewidth}
\includegraphics[width=\linewidth]{figures/la_benchmark_2_time_vs_size.pdf}
\caption{Execution time vs. matrix dimension (linear scale).}
\label{fig:laplotdense-linear}
\end{subfigure}
\hfill
\begin{subfigure}{0.48\linewidth}
\includegraphics[width=\linewidth]{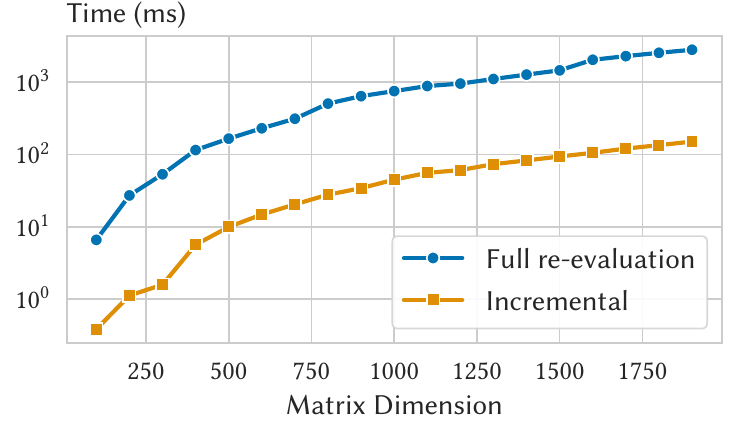}
\caption{Execution time vs. matrix dimension (logarithmic scale).}
\label{fig:laplotdense-log}
\end{subfigure}
\caption{Dense layer (matrix-vector multiplication with bias and ReLU activation) benchmark execution time vs. matrix dimension.}
\label{fig:laplotdense}
\end{figure}

To understand how the sparsity of the input affects performance, we also benchmark matrix-vector multiplication with varying sparsity levels (percentage of non-zero elements) on fixed $1000 \times 1000$ matrices.
The results are plotted in Fig.~\ref{fig:laplotsparsity}.

\begin{figure}
\centering
\includegraphics[width=0.6\linewidth]{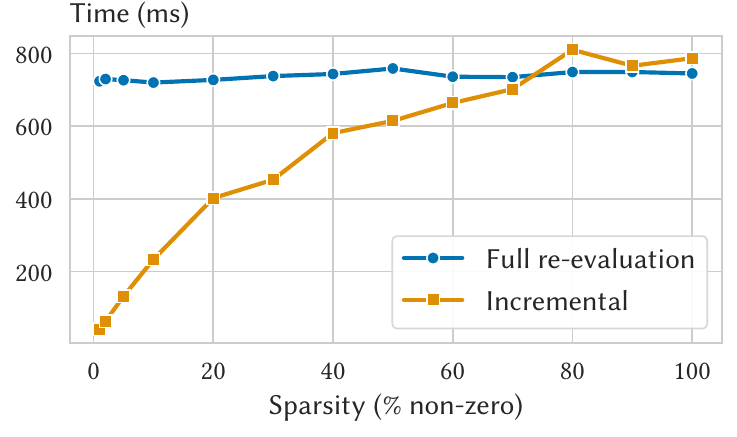}
\caption{Matrix-vector multiplication benchmark execution time vs. input sparsity (percentage of non-zero elements), with fixed matrix dimension $1000 \times 1000$.}
\label{fig:laplotsparsity}
\end{figure}

\subsection{Relational Algebra}
\label{sec:raeval}

We evaluate two relational algebra operations: an aggregation operation (projection) and a non-aggregation operation (join, then filter).

\paragraph{Aggregation}

In Fig.~\ref{fig:groupby}, we show the library implementation
for our relational algebra aggregation case study (note that \lstinline{Z} is
the relation type constructor). For both \lstinline{count} and \lstinline{group},
we write a direct implementation in Lean and incrementalizing it with the \lstinline{Incr.self}
combinator. We then define a derived operation \lstinline{proj}, which
projects the input relation on an input function \lstinline{f}.

The results are plotted in Fig.~\ref{fig:raplotagg}. The change size is 1\% of the input size.

\begin{figure}
\begin{lstlisting}
def countD (m : Z α) : Int := m.fold (λ c _ v => c + v) 0
def count : Incr (Z α) Int := .self countD countD

def groupByD (f : α → β) (m : Z α) : D β (Z α) := Id.run do
  let mut res := ∅
  for (k, v) in m do
    res := res.alter (f k) (λ
      | .some m => m.add k v
      | .none => some {(k, v)})
  return res

def groupBy (f : α → β) : Incr (Z α) (D β (Z α)) :=
  Incr.self (groupByD f) (groupByD f)

def proj (f : α → β) : Incr (D α Int) (D β Int) :=
  groupBy f ;; Incr.map count
\end{lstlisting}
\caption{Implementation of \lstinline{count}, \lstinline{groupBy} and \lstinline{proj}.}
\label{fig:groupby}
\end{figure}

\begin{figure}
\begin{subfigure}{0.48\linewidth}
\includegraphics[width=\linewidth]{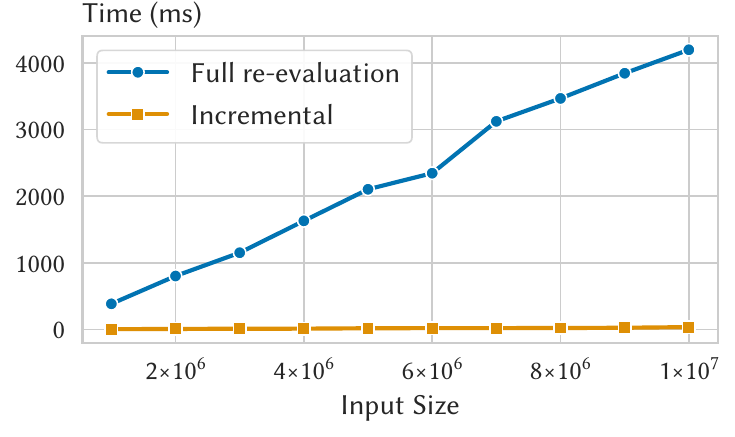}
\caption{Execution time vs. input size (linear scale).}
\label{fig:raplotagg-linear}
\end{subfigure}
\hfill
\begin{subfigure}{0.48\linewidth}
\includegraphics[width=\linewidth]{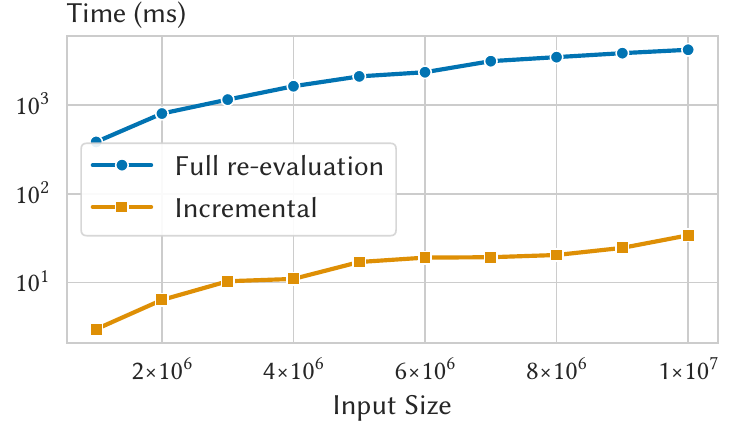}
\caption{Execution time vs. input size (logarithmic scale).}
\label{fig:raplotagg-log}
\end{subfigure}
\caption{Relational Algebra aggregation (projection) benchmark execution time vs. input size.}
\label{fig:raplotagg}
\end{figure}

\paragraph{Join}

For non-aggregation operations, we benchmark
the join $ R \bowtie_p S := \ext{\otimes}\src{; filter~}p $.
The results are plotted in Fig.~\ref{fig:raplotjoin}. The change size is 1\% of the input size.

\begin{figure}
\begin{subfigure}{0.48\linewidth}
\includegraphics[width=\linewidth]{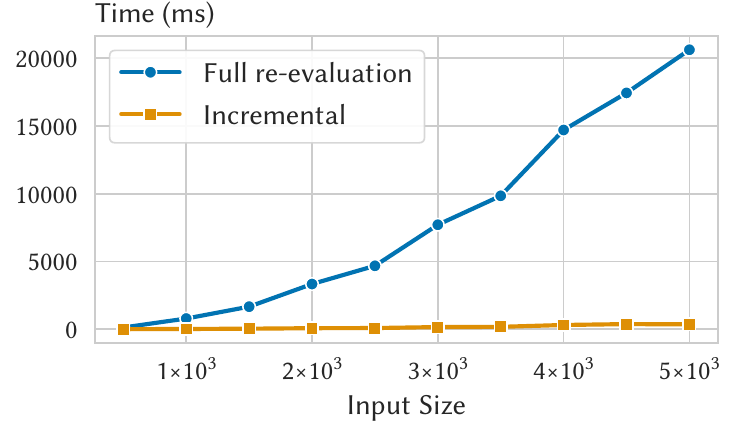}
\caption{Execution time vs. input size (linear scale).}
\label{fig:raplotjoin-linear}
\end{subfigure}
\hfill
\begin{subfigure}{0.48\linewidth}
\includegraphics[width=\linewidth]{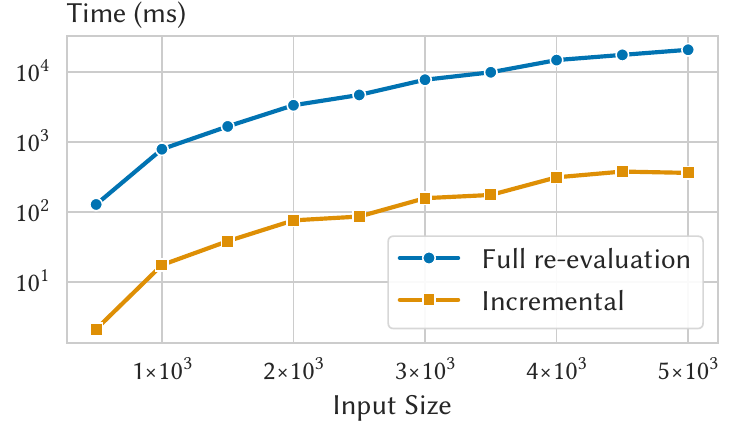}
\caption{Execution time vs. input size (logarithmic scale).}
\label{fig:raplotjoin-log}
\end{subfigure}
\caption{Relational Algebra join benchmark execution time vs. input size.}
\label{fig:raplotjoin}
\end{figure}

\subsection{JSON}
\label{sec:jsoneval}

In Fig.~\ref{fig:data} we show the example data used in our JSON case study and in Fig.~\ref{fig:q1} we show the library implementation for the Q1 query.
The program first removes all books that are not published by Addison-Wesley, then removes all books not published after 1991, and finally projects the result on the \lstinline{year} and \lstinline{title} fields.
The helper function \lstinline{checkPath} takes a path and a predicate and returns an incremental function which deletes JSON values that do not satisfy the predicate at the given path.

\begin{figure}
\begin{lstlisting}
structure Person where
  first : String
  last : String
  affiliation : Option String := none

structure Book where
  year : Nat
  title : String
  authors : List Person := []
  editors : List Person := []
  publisher : String
  price : Float

def bibliography : Array Book := #[
  { year := 1994,
    title := "TCP/IP Illustrated",
    authors := [{ first := "W.", last := "Stevens" }],
    publisher := "Addison-Wesley",
    price := 65.95 },
  { year := 1992,
    title := "Advanced Programming in the Unix environment",
    authors := [{ first := "W.", last := "Stevens" }],
    publisher := "Addison-Wesley",
    price := 65.95 },
  { year := 2000,
    title := "Data on the Web",
    authors := [
      { first := "Serge", last := "Abiteboul" },
      { first := "Peter", last := "Buneman" },
      { first := "Dan", last := "Suciu" }
    ],
    publisher := "Morgan Kaufmann Publishers",
    price := 39.95 },
  { year := 1999,
    title := "The Economics of Technology and Content for Digital TV",
    editors := [{ first := "Darcy", last := "Gerbarg", affiliation := "CITI" }],
    publisher := "Kluwer Academic Publishers",
    price := 129.95 }
]
\end{lstlisting}
\caption{Test data for Q1 from \url{https://www.w3.org/TR/xquery-use-cases/}.
Lean can automatically derive JSON serialization and deserialization ; we further convert the JSON to our container representation (both omitted here).}
\label{fig:data}
\end{figure}

\begin{figure}
\begin{lstlisting}
def checkPath (p : JsonPath) (q : JsonScalar → Bool) : Incr Json Json :=
  .triv λ v => if q (v.get p) then v else ∅

def q1I : Incr (D Nat Json) (D Nat Json) :=
  let f1 : Incr (D Nat Json) (D Nat Json) :=
    .map (checkPath [.dict "publisher"] λ v =>
      v == (.str "Addison-Wesley"))
  let f2 : Incr (D Nat Json) (D Nat Json) :=
    .map (checkPath [.dict "year"] λ
      | .num n => n > 1991
      | _ => false)
  f1 ;; f2 ;; .map (.filter λ k => k == [.dict "year"] || k == [.dict "title"])
\end{lstlisting}
\caption{Incremental implementation of Q1. ``\lstinline{D}'' is a dictionary type,
``\lstinline{Json}'' is the JSON container.}
\label{fig:q1}
\end{figure}

\subsection{Rose Trees}
\label{sec:roseeval}

For the rose trees evaluation, we generate trees with branching factor 2,
and define an incremental sum over the tree.
The change size is 1\% of the tree size.
The results are plotted in Fig.~\ref{fig:roseplot}.

\begin{figure}
\begin{subfigure}{0.48\linewidth}
\includegraphics[width=\linewidth]{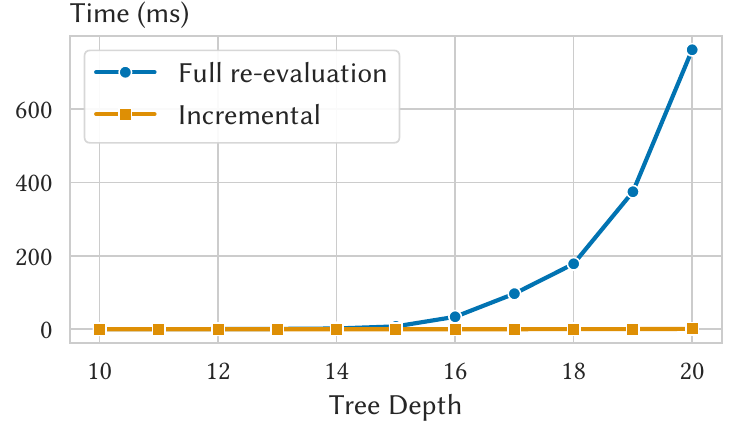}
\caption{Execution time vs. tree depth (linear scale).}
\label{fig:roseplot-linear}
\end{subfigure}
\hfill
\begin{subfigure}{0.48\linewidth}
\includegraphics[width=\linewidth]{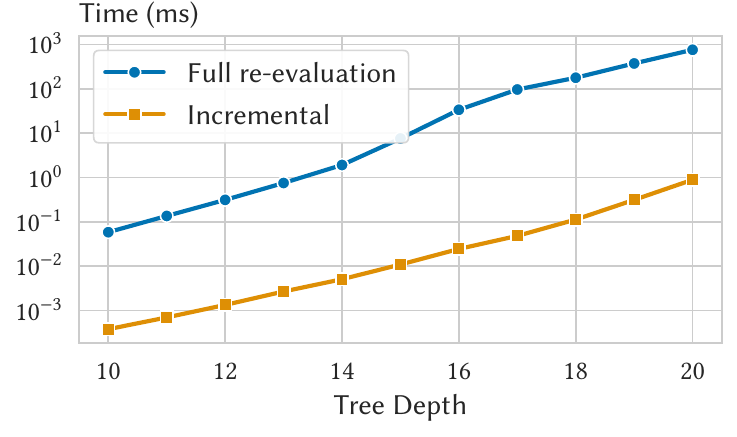}
\caption{Execution time vs. tree depth (logarithmic scale).}
\label{fig:roseplot-log}
\end{subfigure}
\caption{Rose tree benchmark execution time vs. tree depth.}
\label{fig:roseplot}
\end{figure}

%% file: bib.bib
@inproceedings{Cai14,
  author       = {Yufei Cai and
                  Paolo G. Giarrusso and
                  Tillmann Rendel and
                  Klaus Ostermann},
  editor       = {Michael F. P. O'Boyle and
                  Keshav Pingali},
  title        = {A theory of changes for higher-order languages: incrementalizing {\(\lambda\)}-calculi
                  by static differentiation},
  booktitle    = {{ACM} {SIGPLAN} Conference on Programming Language Design and Implementation,
                  {PLDI} '14, Edinburgh, United Kingdom - June 09 - 11, 2014},
  pages        = {145--155},
  publisher    = {{ACM}},
  year         = {2014},
  url          = {https://doi.org/10.1145/2594291.2594304},
  doi          = {10.1145/2594291.2594304},
  bibsource    = {dblp computer science bibliography, https://dblp.org}
}

@article{Budiu23,
  author       = {Mihai Budiu and
                  Tej Chajed and
                  Frank McSherry and
                  Leonid Ryzhyk and
                  Val Tannen},
  title        = {{DBSP:} Automatic Incremental View Maintenance for Rich Query Languages},
  journal      = {Proc. {VLDB} Endow.},
  volume       = {16},
  number       = {7},
  pages        = {1601--1614},
  year         = {2023},
  url          = {https://www.vldb.org/pvldb/vol16/p1601-budiu.pdf},
  doi          = {10.14778/3587136.3587137},
  bibsource    = {dblp computer science bibliography, https://dblp.org}
}

@article{Abbott05,
  author       = {Michael Gordon Abbott and
                  Thorsten Altenkirch and
                  Neil Ghani},
  title        = {Containers: Constructing strictly positive types},
  journal      = {Theor. Comput. Sci.},
  volume       = {342},
  number       = {1},
  pages        = {3--27},
  year         = {2005},
  url          = {https://doi.org/10.1016/j.tcs.2005.06.002},
  doi          = {10.1016/J.TCS.2005.06.002},
  bibsource    = {dblp computer science bibliography, https://dblp.org}
}

@inproceedings{Giarrusso19,
  author       = {Paolo G. Giarrusso and
                  Yann R{\'{e}}gis{-}Gianas and
                  Philipp Schuster},
  editor       = {Lu{\'{\i}}s Caires},
  title        = {Incremental $\lambda$-Calculus in Cache-Transfer Style
                  - Static Memoization by Program Transformation},
  booktitle    = {Programming Languages and Systems - 28th European Symposium on Programming,
                  {ESOP} 2019, Held as Part of the European Joint Conferences on Theory
                  and Practice of Software, {ETAPS} 2019, Prague, Czech Republic, April
                  6-11, 2019, Proceedings},
  series       = {Lecture Notes in Computer Science},
  volume       = {11423},
  pages        = {553--580},
  publisher    = {Springer},
  year         = {2019},
  url          = {https://doi.org/10.1007/978-3-030-17184-1\_20},
  doi          = {10.1007/978-3-030-17184-1\_20},
  bibsource    = {dblp computer science bibliography, https://dblp.org}
}

@article{Morihata18,
  author       = {Akimasa Morihata},
  title        = {Incremental computing with data structures},
  journal      = {Sci. Comput. Program.},
  volume       = {164},
  pages        = {18--36},
  year         = {2018},
  url          = {https://doi.org/10.1016/j.scico.2017.04.001},
  doi          = {10.1016/J.SCICO.2017.04.001},
  bibsource    = {dblp computer science bibliography, https://dblp.org}
}

@article{Shaikhha20,
  author       = {Amir Shaikhha and
                  Mohammed Elseidy and
                  Stephan Mihaila and
                  Daniel Espino and
                  Christoph Koch},
  title        = {Synthesis of Incremental Linear Algebra Programs},
  journal      = {{ACM} Trans. Database Syst.},
  volume       = {45},
  number       = {3},
  pages        = {12:1--12:44},
  year         = {2020},
  url          = {https://doi.org/10.1145/3385398},
  doi          = {10.1145/3385398},
  bibsource    = {dblp computer science bibliography, https://dblp.org}
}

@inproceedings{Bohler24,
  author       = {Timon B{\"{o}}hler and
                  David Richter and
                  Mira Mezini},
  editor       = {Luca Di Stefano},
  title        = {Incrementalizing Polynomial Functors},
  booktitle    = {Proceedings of the 26th {ACM} International Workshop on Formal Techniques
                  for Java-like Programs, FTfJP 2024, Vienna, Austria, 20 September
                  2024},
  pages        = {44--49},
  publisher    = {{ACM}},
  year         = {2024},
  url          = {https://doi.org/10.1145/3678721.3686231},
  doi          = {10.1145/3678721.3686231},
  bibsource    = {dblp computer science bibliography, https://dblp.org}
}

@article{Shaikhha22,
  author       = {Amir Shaikhha and
                  Mathieu Huot and
                  Jaclyn Smith and
                  Dan Olteanu},
  title        = {Functional collection programming with semi-ring dictionaries},
  journal      = {Proc. {ACM} Program. Lang.},
  volume       = {6},
  number       = {{OOPSLA1}},
  pages        = {1--33},
  year         = {2022},
  url          = {https://doi.org/10.1145/3527333},
  doi          = {10.1145/3527333},
  bibsource    = {dblp computer science bibliography, https://dblp.org}
}

@inproceedings{Hammer14,
  author       = {Matthew A. Hammer and
                  Yit Phang Khoo and
                  Michael Hicks and
                  Jeffrey S. Foster},
  editor       = {Michael F. P. O'Boyle and
                  Keshav Pingali},
  title        = {Adapton: composable, demand-driven incremental computation},
  booktitle    = {{ACM} {SIGPLAN} Conference on Programming Language Design and Implementation,
                  {PLDI} '14, Edinburgh, United Kingdom - June 09 - 11, 2014},
  pages        = {156--166},
  publisher    = {{ACM}},
  year         = {2014},
  url          = {https://doi.org/10.1145/2594291.2594324},
  doi          = {10.1145/2594291.2594324},
  bibsource    = {dblp computer science bibliography, https://dblp.org}
}

@article{Elliott17,
  author       = {Conal Elliott},
  title        = {Compiling to categories},
  journal      = {Proc. {ACM} Program. Lang.},
  volume       = {1},
  number       = {{ICFP}},
  pages        = {27:1--27:27},
  year         = {2017},
  url          = {https://doi.org/10.1145/3110271},
  doi          = {10.1145/3110271},
  bibsource    = {dblp computer science bibliography, https://dblp.org}
}

@article{Carette09,
  author       = {Jacques Carette and
                  Oleg Kiselyov and
                  Chung{-}chieh Shan},
  title        = {Finally tagless, partially evaluated: Tagless staged interpreters
                  for simpler typed languages},
  journal      = {J. Funct. Program.},
  volume       = {19},
  number       = {5},
  pages        = {509--543},
  year         = {2009},
  url          = {https://doi.org/10.1017/S0956796809007205},
  doi          = {10.1017/S0956796809007205},
  bibsource    = {dblp computer science bibliography, https://dblp.org}
}

@article{Liu00,
  author       = {Yanhong A. Liu},
  title        = {Efficiency by Incrementalization: An Introduction},
  journal      = {High. Order Symb. Comput.},
  volume       = {13},
  number       = {4},
  pages        = {289--313},
  year         = {2000},
  url          = {https://doi.org/10.1023/A:1026547031739},
  doi          = {10.1023/A:1026547031739},
  bibsource    = {dblp computer science bibliography, https://dblp.org}
}

@article{Patrignani20,
  title={Why Should Anyone use Colours? or, Syntax Highlighting Beyond Code Snippets},
  author={Marco Patrignani},
  journal={ArXiv},
  year={2020},
  volume={abs/2001.11334},
}

@article{Matsuda23,
  author       = {Kazutaka Matsuda and
                  Samantha Frohlich and
                  Meng Wang and
                  Nicolas Wu},
  title        = {Embedding by Unembedding},
  journal      = {Proc. {ACM} Program. Lang.},
  volume       = {7},
  number       = {{ICFP}},
  pages        = {1--47},
  year         = {2023},
  url          = {https://doi.org/10.1145/3607830},
  doi          = {10.1145/3607830},
  bibsource    = {dblp computer science bibliography, https://dblp.org}
}

@article{Almeida25,
  author       = {Paulo S{\'{e}}rgio Almeida},
  title        = {Approaches to Conflict-free Replicated Data Types},
  journal      = {{ACM} Comput. Surv.},
  volume       = {57},
  number       = {2},
  pages        = {51:1--51:36},
  year         = {2025},
  url          = {https://doi.org/10.1145/3695249},
  doi          = {10.1145/3695249},
  bibsource    = {dblp computer science bibliography, https://dblp.org}
}

@inproceedings{Shapiro11,
  author       = {Marc Shapiro and
                  Nuno M. Pregui{\c{c}}a and
                  Carlos Baquero and
                  Marek Zawirski},
  editor       = {Xavier D{\'{e}}fago and
                  Franck Petit and
                  Vincent Villain},
  title        = {Conflict-Free Replicated Data Types},
  booktitle    = {Stabilization, Safety, and Security of Distributed Systems - 13th
                  International Symposium, {SSS} 2011, Grenoble, France, October 10-12,
                  2011. Proceedings},
  series       = {Lecture Notes in Computer Science},
  volume       = {6976},
  pages        = {386--400},
  publisher    = {Springer},
  year         = {2011},
  url          = {https://doi.org/10.1007/978-3-642-24550-3\_29},
  doi          = {10.1007/978-3-642-24550-3\_29},
  bibsource    = {dblp computer science bibliography, https://dblp.org}
}

@article{Almeida18,
  author       = {Paulo S{\'{e}}rgio Almeida and
                  Ali Shoker and
                  Carlos Baquero},
  title        = {Delta state replicated data types},
  journal      = {J. Parallel Distributed Comput.},
  volume       = {111},
  pages        = {162--173},
  year         = {2018},
  url          = {https://doi.org/10.1016/j.jpdc.2017.08.003},
  doi          = {10.1016/J.JPDC.2017.08.003},
  bibsource    = {dblp computer science bibliography, https://dblp.org}
}

@article{Codd70,
  author       = {E. F. Codd},
  title        = {A Relational Model of Data for Large Shared Data Banks},
  journal      = {Commun. {ACM}},
  volume       = {13},
  number       = {6},
  pages        = {377--387},
  year         = {1970},
  url          = {https://doi.org/10.1145/362384.362685},
  doi          = {10.1145/362384.362685},
  bibsource    = {dblp computer science bibliography, https://dblp.org}
}

@inproceedings{Kumar25,
  author       = {Prashant Kumar and
                  Andr{\'{e}} Pacak and
                  Sebastian Erdweg},
  editor       = {Jonathan Aldrich and
                  Alexandra Silva},
  title        = {Incremental Computing by Differential Execution},
  booktitle    = {39th European Conference on Object-Oriented Programming, {ECOOP} 2025,
                  June 30 to July 2, 2025, Bergen, Norway},
  series       = {LIPIcs},
  volume       = {333},
  pages        = {20:1--20:24},
  publisher    = {Schloss Dagstuhl - Leibniz-Zentrum f{\"{u}}r Informatik},
  year         = {2025},
  url          = {https://doi.org/10.4230/LIPIcs.ECOOP.2025.20},
  doi          = {10.4230/LIPICS.ECOOP.2025.20},
  bibsource    = {dblp computer science bibliography, https://dblp.org}
}

@inproceedings{Reinking20,
author = {Reinking, Alex and Xie, Ningning and de Moura, Leonardo and Leijen, Daan},
title = {Perceus: garbage free reference counting with reuse},
year = {2021},
isbn = {9781450383912},
publisher = {Association for Computing Machinery},
address = {New York, NY, USA},
url = {https://doi.org/10.1145/3453483.3454032},
doi = {10.1145/3453483.3454032},
booktitle = {Proceedings of the 42nd ACM SIGPLAN International Conference on Programming Language Design and Implementation},
pages = {96–111},
numpages = {16},
location = {Virtual, Canada},
series = {PLDI 2021}
}

@book{Wagner14,
  title={Symmetric edit lenses: a new foundation for bidirectional languages},
  author={Wagner, Daniel},
  year={2014},
  publisher={University of Pennsylvania}
}

@inproceedings{McSherry13,
  author       = {Frank McSherry and
                  Derek Gordon Murray and
                  Rebecca Isaacs and
                  Michael Isard},
  title        = {Differential Dataflow},
  booktitle    = {Sixth Biennial Conference on Innovative Data Systems Research, {CIDR}
                  2013, Asilomar, CA, USA, January 6-9, 2013, Online Proceedings},
  publisher    = {www.cidrdb.org},
  year         = {2013},
  url          = {http://cidrdb.org/cidr2013/Papers/CIDR13\_Paper111.pdf},
  bibsource    = {dblp computer science bibliography, https://dblp.org}
}

@misc{Ghica25,
      title={A Complete Theory of Sequential Digital Circuits: Denotational, Operational and Algebraic Semantics}, 
      author={Dan R. Ghica and George Kaye and David Sprunger},
      year={2025},
      eprint={2201.10456},
      archivePrefix={arXiv},
      primaryClass={cs.LO},
      url={https://arxiv.org/abs/2201.10456}, 
}

@inproceedings{Lean4,
	author       = {Leonardo de Moura and
	Sebastian Ullrich},
	editor       = {Andr{\'{e}} Platzer and
	Geoff Sutcliffe},
	title        = {The Lean 4 Theorem Prover and Programming Language},
	booktitle    = {Automated Deduction - {CADE} 28 - 28th International Conference on
	Automated Deduction, Virtual Event, July 12-15, 2021, Proceedings},
	series       = {Lecture Notes in Computer Science},
	volume       = {12699},
	pages        = {625--635},
	publisher    = {Springer},
	year         = {2021},
	url          = {https://doi.org/10.1007/978-3-030-79876-5\_37},
	doi          = {10.1007/978-3-030-79876-5\_37},
	bibsource    = {dblp computer science bibliography, https://dblp.org}
}

@book{Okasaki99,
  author       = {Chris Okasaki},
  title        = {Purely functional data structures},
  publisher    = {Cambridge University Press},
  year         = {1999},
  isbn         = {978-0-521-66350-2},
  bibsource    = {dblp computer science bibliography, https://dblp.org}
}

@article{Lindley10,
  author       = {Sam Lindley and
                  Philip Wadler and
                  Jeremy Yallop},
  title        = {The arrow calculus},
  journal      = {J. Funct. Program.},
  volume       = {20},
  number       = {1},
  pages        = {51--69},
  year         = {2010},
  url          = {https://doi.org/10.1017/S095679680999027X},
  doi          = {10.1017/S095679680999027X},
  bibsource    = {dblp computer science bibliography, https://dblp.org}
}

@article{Laddad25,
  author       = {Shadaj Laddad and
                  Alvin Cheung and
                  Joseph M. Hellerstein and
                  Mae Milano},
  title        = {Flo: {A} Semantic Foundation for Progressive Stream Processing},
  journal      = {Proc. {ACM} Program. Lang.},
  volume       = {9},
  number       = {{POPL}},
  pages        = {241--270},
  year         = {2025},
  url          = {https://doi.org/10.1145/3704845},
  doi          = {10.1145/3704845},
  bibsource    = {dblp computer science bibliography, https://dblp.org}
}

@inproceedings{Meiklejohn15,
  author       = {Christopher Meiklejohn and
                  Peter Van Roy},
  editor       = {Moreno Falaschi and
                  Elvira Albert},
  title        = {Lasp: a language for distributed, coordination-free programming},
  booktitle    = {Proceedings of the 17th International Symposium on Principles and
                  Practice of Declarative Programming, Siena, Italy, July 14-16, 2015},
  pages        = {184--195},
  publisher    = {{ACM}},
  year         = {2015},
  url          = {https://doi.org/10.1145/2790449.2790525},
  doi          = {10.1145/2790449.2790525},
  bibsource    = {dblp computer science bibliography, https://dblp.org}
}

@inproceedings{Conway12,
  author       = {Neil Conway and
                  William R. Marczak and
                  Peter Alvaro and
                  Joseph M. Hellerstein and
                  David Maier},
  editor       = {Michael J. Carey and
                  Steven Hand},
  title        = {Logic and lattices for distributed programming},
  booktitle    = {{ACM} Symposium on Cloud Computing, {SOCC} '12, San Jose, CA, USA,
                  October 14-17, 2012},
  pages        = {1},
  publisher    = {{ACM}},
  year         = {2012},
  url          = {https://doi.org/10.1145/2391229.2391230},
  doi          = {10.1145/2391229.2391230},
  bibsource    = {dblp computer science bibliography, https://dblp.org}
}

@inproceedings{Milano19,
  author       = {Mae Milano and
                  Rolph Recto and
                  Tom Magrino and
                  Andrew C. Myers},
  editor       = {Benjamin S. Lerner and
                  Rastislav Bod{\'{\i}}k and
                  Shriram Krishnamurthi},
  title        = {A Tour of Gallifrey, a Language for Geodistributed Programming},
  booktitle    = {3rd Summit on Advances in Programming Languages, {SNAPL} 2019, Providence,
                  RI, USA, May 16-17, 2019},
  series       = {LIPIcs},
  volume       = {136},
  pages        = {11:1--11:19},
  publisher    = {Schloss Dagstuhl - Leibniz-Zentrum f{\"{u}}r Informatik},
  year         = {2019},
  url          = {https://doi.org/10.4230/LIPIcs.SNAPL.2019.11},
  doi          = {10.4230/LIPICS.SNAPL.2019.11},
  bibsource    = {dblp computer science bibliography, https://dblp.org}
}

@inproceedings{Kuper13,
author = {Kuper, Lindsey and Newton, Ryan R.},
title = {LVars: lattice-based data structures for deterministic parallelism},
year = {2013},
isbn = {9781450323819},
publisher = {Association for Computing Machinery},
address = {New York, NY, USA},
url = {https://doi.org/10.1145/2502323.2502326},
doi = {10.1145/2502323.2502326},
booktitle = {Proceedings of the 2nd ACM SIGPLAN Workshop on Functional High-Performance Computing},
pages = {71–84},
numpages = {14},
location = {Boston, Massachusetts, USA},
series = {FHPC '13}
}

@article{Smeding24,
  author       = {Tom Smeding and
                  Matthijs V{\'{a}}k{\'{a}}r},
  title        = {Efficient {CHAD}},
  journal      = {Proc. {ACM} Program. Lang.},
  volume       = {8},
  number       = {{POPL}},
  pages        = {1060--1088},
  year         = {2024},
  url          = {https://doi.org/10.1145/3632878},
  doi          = {10.1145/3632878},
  bibsource    = {dblp computer science bibliography, https://dblp.org}
}

@software{Boehler26,
  author       = {Böhler, Timon and
                  Reinhard, Tobias and
                  Richter, David and
                  Mezini, Mira},
  title        = {DeCo: A Core Calculus for Incremental Functional
                   Programming with Generic Data Types
                  },
  month        = feb,
  year         = 2026,
  publisher    = {Zenodo},
  doi          = {10.5281/zenodo.18757667},
  url          = {https://doi.org/10.5281/zenodo.18757667},
}
